\documentclass[a4paper,11pt]{scrartcl}
\usepackage{graphicx} 
\usepackage[pdfpagelabels]{hyperref}
\usepackage[utf8]{inputenc}
\usepackage{amsmath,amsfonts,amssymb,amsthm,mathtools}
\usepackage{amsmath,verbatim}
\usepackage{amsfonts}
\usepackage{amssymb}
\usepackage{dsfont}
\usepackage{mathrsfs}
\usepackage[T1]{fontenc}
\usepackage{lmodern}
\usepackage[ngerman,english]{babel}
\usepackage[inline]{enumitem}
\usepackage{latexsym,graphicx}
\usepackage[all]{xy}
\usepackage{mathtools}
\usepackage{color}
\usepackage{authblk}
\usepackage{amscd}
\usepackage{microtype}
\usepackage{mathrsfs,cite}
\usepackage{upgreek}
\numberwithin{equation}{section}
\usepackage{txfonts}
\usepackage{csquotes}

\newtheorem{theorem}{Theorem}
\newtheorem{proposition}{Proposition}[section]
\newtheorem{lemma}[proposition]{Lemma}
\theoremstyle{definition}
\newtheorem{remark}[proposition]{Remark}

\newtheorem{assumption}{Assumptions}

\title{Large Deviations for the Bose--Einstein Condensate of the Ideal Gas in the Canonical Ensemble}
\author{Andreas Deuchert, Xuanyu Li}
\date{}

\begin{document}

\maketitle
\begin{abstract}
We consider a non-interacting Bose gas governed by a general one-particle Hamiltonian in the canonical ensemble. Our main results are large-deviation estimates for the number of particles in the Bose--Einstein condensate. We show that the left- and right-tail probabilities decay at different exponential rates as the number of particles tends to infinity. 
\end{abstract}
\tableofcontents

\section{Introduction and Main Results}
\subsection{Background}
The first experimental realization of Bose--Einstein condensation in dilute alkali gases in 1995 \cite{Andetal1995,DavisEtal1995}
triggered a wealth of mathematical research on Bose gases, see, e.g.,
the books and review articles \cite{LieSeiSolYng2005,Solovej2025,Roug2015,Nap2023,BenPoSchl2016}.
For recent developments in the analysis of Bose gases at positive temperature that are not covered in these surveys, we refer to
\cite{DeuSeiYng2019,DeuSei2020,LewNamRou2021,FroeKnowSchlSoh2022,DeuCapSchl2024,BocDeuSto2024,CapDeu2025,DeuNamNap2025,FroeKnowSchlSoh2025,NamZhuZhu2025,BastBocCenaDeu2025}. 

In early experiments on Bose--Einstein condensates, only a small number of macroscopic observables, such as the condensate fraction, could be measured. Since then, experimental techniques have improved significantly, enabling increasingly precise investigations of the properties of quantum gases. We mention a few selected examples. In \cite{SakKase2016}, the difference between the empirical measure and the density of the system is quantified, and finite-size effects are identified. In \cite{Christetal2021}, the authors observe a discrepancy between the measured scaling of condensate-number fluctuations with the total particle number and the predictions of the canonical ensemble. This observation raises the question of which statistical ensemble provides the most appropriate description of a cold alkali gas following laser and evaporative cooling. The authors of \cite{HerceEtAl2023} measure the full distribution, or full counting statistics, of the number of particles contained in certain subsets of configuration space. In \cite{TenartEtAl2021,LamiraultEtAl2025}, correlations between the occupation numbers of different momentum modes in weakly interacting gases are investigated. Finally, the measurements reported in \cite{Bureikatal2023} reveal that Bogoliubov pairing between modes of opposite momenta, which is characteristic of the weakly interacting regime, is suppressed as the interaction strength increases.
 
Motivated in part by these developments, we study in this article the large-deviation behavior of the number of particles $N_0$ in the Bose--Einstein condensate of an ideal, i.e., non-interacting, Bose gas in the canonical ensemble. In other words, we study the exponential decay of the probability that $N_0$ deviates from its mean by an amount that grows proportionally to the total particle number $N$ as $N \to \infty$. Non-interacting systems are physically relevant because the interaction strength between the particles can be controlled experimentally and is weak in many situations. We hope that our predictions can be tested in future experiments on cold alkali gases. 

A further motivation is mathematical. The limiting distribution of a centered and appropriately rescaled version of $N_0$ was studied in \cite{ChatDiac2014}. The authors obtain both Gaussian and non-Gaussian limiting distributions, whose variances grow at different rates as $N \to \infty$. The growth rate of the variance depends on the eigenvalue asymptotics of the underlying one-particle Hamiltonian $H$. These results should be compared with those for the exactly solvable grand canonical ensemble, in which $N_0$ always follows a geometric distribution. We emphasize that the canonical ensemble is not exactly solvable. This naturally raises the question of whether the large-deviation behavior of the model can nevertheless be determined. We answer this question affirmatively. Our results show that for all Hamiltonians $H$ under consideration, the left- and right-tail probabilities
\begin{equation}
\mathbb{P}\bigl[N_0<(1-\eta)\mathbb{E}(N_0)\bigr]
\qquad\text{and}\qquad
\mathbb{P}\bigl[N_0>(1+\eta)\mathbb{E}(N_0)\bigr]
\end{equation}
decay at different exponential rates as $N \to \infty$. Here $0 < \eta < 1$ in the first case and $\eta > 1$ in the second case. This asymmetry should be contrasted with the results of \cite{ChatDiac2014}, where some of the limiting distributions of $N_0$ are symmetric.

Large-deviation theory has been applied to bosonic many-particle systems in \cite{BergLewisPule1988}, where the Huang--Yang--Luttinger model is studied, and in \cite{BruZag2008}, which treats the superstable weakly imperfect Bose gas. In \cite{LebLenciSpohn2000}, the authors consider ideal Bose and Fermi gases and investigate large deviations of the number of particles contained in a subregion of the full domain. The same observable is studied for rarefied interacting quantum gases in \cite{GalLebMastr2002}. Tail bounds for interacting Bose--Einstein condensates at zero temperature have been established in \cite{KirkRadeSchlein2021,RaSei-2022,NamRade2025,BehrBreRade2025}. Moreover, \cite{Radem2025} contains a computation of the limiting moment generating function of the number of particles outside the condensate in the Gross--Pitaevskii regime at zero temperature.

In the next section we introduce our mathematical setup. 
\subsection{Mathematical Setup}
\label{sec:mathematicalSetup}
We consider a gas of non-interacting bosons described by the one-particle Hamiltonian $H$. 
\begin{assumption}
\label{Assumption}
We assume that $H$ has discrete spectrum $0 = E_0 < E_1 \leq E_2 \leq \ldots$. Regarding the growth of the eigenvalues, we assume that there exist constants $L>0$ and $\alpha \in [1,\infty)$ such that
\begin{equation}
\lim_{\lambda\to\infty}\frac{\#\{j:\,E_j\le \lambda\}}{L\,\lambda^{\alpha}}=1.
\label{eq:energyAsymptotics}
\end{equation}
\end{assumption}
The strict inequality $E_0 < E_1$ means that the ground state is non-degenerate. A prominent class of examples is provided by Schrödinger operators of the form $-\Delta + V(x)$ in $\mathbb{R}^d$, where $V : \mathbb{R}^d \to \mathbb{R}$ is sufficiently regular and satisfies $\lim_{|x| \to \infty} V(x)/(a |x|^s) = 1$ for some $a,s > 0$. In this case 
\begin{equation}
    L = \frac{a^{-d/s} \int_0^1 (1-x^s)^{d/2} x^{d-1} \mathrm{d} x}{2^{d-1} \Gamma(d/2+1) \Gamma(d/2)}  \quad \text{ and } \quad \alpha = \frac{2d+ds}{2s},
\end{equation}
where $\Gamma$ denotes the Gamma function. Note that we require $\frac{2d+ds}{2s} \geq 1$. Another class of examples is provided by Schrödinger operators $-\Delta + V(x)$ on a bounded domain $\Omega$ with sufficiently regular boundary, equipped with either Dirichlet or Neumann boundary conditions, in dimensions $d \geq 2$. In this case 
\begin{equation}
    L = \frac{| \Omega |}{2^d \pi^{d/2} \Gamma(d/2+1)} \quad \text{ and } \quad \alpha = d/2,
\end{equation}
where $| \Omega |$ denotes the Lebesgue measure of $\Omega$. Both cases are consequences of Weyl's law for the relevant operators, see, e.g., \cite{RS4}.

We call a vector  $n=(n_0,n_1,n_2,...)$ of infinite length an occupation number configuration if $n_j \in \mathbb{N}_0$ for all $j \geq 0$ and if only finitely many of the $n_j$ are nonzero. In the canonical ensemble governed by the one-particle Hamiltonian $H$ the system is in the state described by the occupation number configuration $n$ with probability
\begin{equation}
    p_{\beta,N}(n) = \frac{\exp(- \beta \sum_{j=0}^{\infty} E_j n_j) \mathds{1}(|n| = N)}{\sum_{|n|=N} \exp(-\beta \sum_{j=0}^{\infty} E_j n_j)}. 
    \label{eqprobabilityCanonicalEnsemble}
\end{equation}
Here $|n| = \sum_j n_j$ and the parameter $\beta > 0$ denotes the inverse temperature. By $\mathds{1}(|n| = N)$ we denote the indicator function of the set $\{ |n| = N \}$ and the sum in the denominator runs over all possible occupation number configurations $n$ with $|n| = N$. For a set $A$ of occupation number configurations, we define its probability by $\mathbb{P}^{\mathrm{can}}(A) = \sum_{n \in A} p_{\beta,N}(n)$.

Under our assumptions, the above system exhibits a phase transition to a Bose--Einstein condensed phase. More precisely, let $N_j$ be the random variable defined by $N_j(n) = n_j$, and choose 
\begin{equation}
    \beta_N = \begin{cases} b N^{-1/\alpha} & \text{ if } \alpha > 1 \\
    \frac{b \log(N)}{N} & \text{ if } \alpha = 1
    \end{cases}
    \label{eq:beta}
\end{equation}
with some fixed $b \in (0, \infty)$. Then 
\begin{equation}
    \lim_{N \to \infty} \frac{\mathbb{E}^{\mathrm{can}}(N_0)}{N} = \left[ 1 - \left( \frac{b_{\mathrm{c}}}{b} \right)^{\alpha} \right]_+, \quad \text{ where } \quad b_{\mathrm{c}} = [L \alpha \Gamma(\alpha) \zeta(\alpha)]^{1/\alpha},
    \label{eq:BECPhaseTransition}
\end{equation}
if $\alpha > 1$ and $b_{\mathrm{c}} = L$ if $\alpha = 1$, see e.g. \cite[Theorem~1]{ChatDiac2014}. By $\zeta$ we denoted the Riemann zeta function. That is, the expected occupation number of the energy level $E_0$ grows proportionally to $N$ if and only if $b > b_{\mathrm{c}}$. Moreover, one can show that $\mathbb{E}^{\mathrm{can}}(N_j) = o(N)$ for $ j\geq 1$. The critical inverse temperature in \eqref{eq:BECPhaseTransition} is the same as for the grand canonical ensemble.

In \cite{ChatDiac2014}, the authors identified the limiting distribution of a centered and appropriately rescaled version of the random variable $N_0$. Here, we are interested in establishing large-deviation estimates for $N_0$. Our results are stated in the next section.
\subsection{Main Results}
Our main result is captured in the following theorem.

\begin{theorem}
\label{thm:main}
We assume that the energy levels $\{ E_j \}_{j=0}^{\infty}$ satisfy Assumption~\ref{Assumption}, choose $\beta_N$ as in \eqref{eq:beta} with some $b > b_{\mathrm{c}}$, and consider the limit $N \to \infty$. The following statements hold.

\begin{enumerate}[label=(\alph*)]
\item \textbf{Left tail.} Let $\alpha \geq 1$. For every fixed $0<\eta<1$,
\begin{equation}
    \mathbb{P}^{\mathrm{can}}\left(N_0<(1-\eta)\mathbb E[N_0]\right) =
    \exp\left( N \beta_N \eta \Big[\Big(\tfrac{b_{\mathrm{c}}}{b}\Big)^\alpha-1\Big]E_1
    +o(N \beta_N) \right).
    \label{eq:main1}
\end{equation}
\item \textbf{Right tail, part I.}
Let $\alpha > 1$. For every fixed $0 < \eta < [(b/b_{\mathrm{c}})^\alpha - 1]^{-1}$, 
\begin{equation}
    \mathbb P^{\mathrm{can}}\left(N_0>(1+\eta)\mathbb E[N_0]\right) = \exp\left( N\inf_{s<0}\phi_{\eta,b}(s)+o(N) \right),
    \label{eq:main2}
\end{equation}
where
\begin{equation}
\phi_{\eta,b}(s) = L\alpha b^{-\alpha} \int_0^\infty x^{\alpha-1}
\log\left(\frac{1-e^{-x}}{1-e^{s-x}}\right) \mathrm{d}x
- \Bigg[ \left(\frac{b_{\mathrm{c}}}{b}\right)^\alpha
- \eta \left(1-\left(\frac{b_{\mathrm{c}}}{b}\right)^\alpha\right) \Bigg]s.
\label{right tail rate function}
\end{equation}    
If $ \eta > [(b/b_{\mathrm{c}})^\alpha - 1]^{-1}$ then $\mathbb P^{\mathrm{can}}(N_0 > (1+\eta) \mathbb{E}[N_0] ) = 0$ for sufficiently large $N$.
\item \textbf{Right tail, part II.} Let $\alpha = 1$ and assume that the limit
\begin{align}
\label{Euler's Constant}
C_{\mathrm{E}} = \lim_{\beta\downarrow0}\left[\frac{\beta}{L}\sum_{j=1}^{\infty}\frac{1}{e^{\beta E_j}-1} -
\log\frac{1}{\beta}\right].
\end{align}
exists. Then, for every fixed $0 < \eta < [b/b_{\mathrm{c}} - 1]^{-1}$,
\begin{equation}
    \mathbb{P}^{\mathrm{can}}\left(N_0>(1+\eta)\mathbb E[N_0]\right)
=\exp\left(-L e^{-(1+\eta)C_{\mathrm{E}}}
\frac{[b\log(b/\beta_N)]^{\eta b/b_{\mathrm{c}}}}{\beta_N^{\eta[b/b_{\mathrm{c}}-1]}}(1+o(1))\right).
    \label{eq:main3}    
\end{equation}
If $\eta > [b/b_{\mathrm{c}} - 1]^{-1}$ then $\mathbb P^{\mathrm{can}}(N_0 > (1+\eta) \mathbb{E}[N_0] ) = 0$ for sufficiently large $N$.
\end{enumerate}
\end{theorem}

Regarding the above theorem we have the following remarks.

\begin{remark}
\begin{enumerate}[label=(\alph*)]
    \item The left- and right-tail probabilities of $N_0$ decay, for every value of $\alpha$, at different exponential rates if $\alpha > 1$. For $\alpha =1 $ the left-tail decays with a power of $N$ while the right-tail displays exponential decay with a rate that depends on $\eta$. This is in marked contrast to the central-limit 
        behavior established in \cite[Theorem~2.2]{ChatDiac2014}, 
        where it is shown that the limiting distribution of $N_0$
        is Gaussian, and hence symmetric, whenever $\alpha \geq 2$. 
        Thus, although typical fluctuations are asymptotically 
        symmetric for $\alpha \geq 2$, the corresponding large 
        deviations remain asymmetric. For $1 \leq \alpha < 2$, the 
        limiting distribution obtained in \cite[Theorem~2.2]{ChatDiac2014} is itself asymmetric.
        \item The case $\alpha=1$ is special from a physical perspective as it describes finite-size Bose--Einstein condensation. As an example, consider a Bose gas in a two-dimensional square box of side length $\ell$ with periodic boundary conditions. This corresponds to the parameter values $\alpha=1$ and $L=\ell^2/(4\pi)$. It is well known that this model does not exhibit Bose--Einstein condensation in the thermodynamic limit $N,\ell\to\infty$ with $N/\ell^2=\varrho$ fixed. However, if the inverse temperature is allowed to diverge logarithmically with $N$ in the thermodynamic limit, a Bose--Einstein condensation phase transition occurs.

        Since the model contains a free parameter, one may equivalently keep the box size fixed and let the inverse temperature tend to zero at the rate $\beta_N \sim \log(N)/N$, see \eqref{eq:beta}. We use a similar rescaling of $\beta$ in the general case $\alpha \geq 1$ as well, which leads to $\beta_N$ in \eqref{eq:beta}.

Experiments with ultracold alkali gases typically involve between $10^2$ and $10^6$ particles, so logarithmic corrections remain moderate. Consequently, finite-size condensation can also be observed in two-dimensional boxes. 

\item Further examples of relevance for experiments include particles confined in a $d$-dimensional harmonic trap ($\alpha=d$, $d\geq 1$) or a $d$-dimensional cubic box with Dirichlet boundary conditions ($\alpha=d/2$, $d \geq 2$). 

\item The asymptotics in parts~(b) and (c) of Theorem~\ref{thm:main} in the case $\eta = [(b/b_{\mathrm{c}})^\alpha - 1]^{-1}$ depend on fine details of the sequence $\{ E_j \}_{j=1}^{\infty}$. Depending on the sign of a certain $N$-dependent expression, the sequence
\begin{equation}
   N \mapsto \mathbb{P}^{\mathrm{can}}\left(N_0>(1+\eta)\mathbb E[N_0] \right) 
\end{equation}
alternates between the two regimes obtained for $\eta < [(b/b_{\mathrm{c}})^\alpha - 1]^{-1}$ and $\eta > [(b/b_{\mathrm{c}})^\alpha - 1]^{-1}$.
\end{enumerate}
\end{remark}
\section{Proofs of the main results}
To simplify notation, we omit the subscript $N$ from $\beta_N$ in \eqref{eq:beta} throughout the proofs. By $C>0$, we denote a generic constant whose value may change from line to line. 
\subsection{Preparations and proof strategy}
\label{sec:preparations}
In this section, we introduce notation and prove two preparatory lemmas in order not to interrupt the main line of the argument later. We also state three propositions, which are proved in Sections~\ref{sec:left tail}, \ref{sec:right tail}, and \ref{sec:right tail1}, and explain our proof strategy.  

Let $\{ Z_j \}_{j=1}^{\infty}$ be an independent family of geometric random variables with law
\begin{equation}
    \mathbb{P}(Z_j=k)=(1-e^{-\beta E_j})e^{-k\beta E_j},
\qquad k=0,1,2,\dots.
\label{eq:defZj}
\end{equation}
We also introduce the notations
\begin{align}
M = \sum_{j=1}^\infty Z_j \quad \text{ and } \quad Z_0 = N-M.
\label{definition of M}
\end{align}
An applications of Kolmogorov's two series theorem and \eqref{eq:energyAsymptotics} show that the above series converges almost surely to a finite value. We highlight that $Z_0$ may assume negative values. The expectations of $Z_j$ and $M$ are given by
\begin{align}
\mathbb{E}[Z_j]=\frac{1}{e^{\beta E_j}-1} \quad \text{ and } \quad \mathbb{E}[M] = \sum_{j=1}^\infty \frac{1}{e^{\beta E_j}-1}.
\label{expectation of Z_j}
\end{align}

The following lemma provides us with a relation between the random variables $\{ Z_j \}_{j=0}^{\infty}$ and the random variables $\{ N_j \}_{j=0}^{\infty}$ defined above \eqref{eq:beta}. Its proof can be found in \cite[Lemma~3.4]{ChatDiac2014}. The lemma is also covered in \cite{PraCara2013}, and it has likely appeared in other parts of the probability and mathematical physics literature before. 
\begin{lemma}
\label{lemma~conditon}
Let $N \geq 1$, $\beta > 0$. For any occupation number configuration $n=(n_0,n_1,n_2,\ldots)$ we have
\begin{equation}
    \mathbb{P}^{\mathrm{can}}( N_0 = n_0, N_1 = n_1, N_2 = n_2, \ldots ) = \mathbb{P}( Z_0 = n_0, Z_1 = n_1, Z_2 = n_2, \ldots  | M \leq N).
\end{equation}
\end{lemma}

The above lemma is interesting because $\{ Z_j \}_{j=1}^{\infty}$ is an independent family of random variables, $Z_0$ is a function of the $Z_j$, $j\geq 1$, and hence all correlations are captured by conditioning on the event $\{ M \leq N \}$. As we will show below, this event has probability close to one in the condensed phase. This is physically natural, since $M$ represents the number of thermally excited particles and Bose--Einstein condensation in the ideal gas occurs because the occupation of the excited energy levels $E_j$, $j\geq1$, becomes saturated. From a mathematical point of view, $N-M$ has positive expectation of order $N$, while the standard deviation of $M$ is much smaller than $N$. It therefore follows that $\mathbb{P}(M\leq N)$ is close to one. Rigorous proofs of these claims can be found in \cite[Lemmas 3.2 and 3.3]{ChatDiac2014}. These observations should be compared with \eqref{eqprobabilityCanonicalEnsemble}, where correlations are introduced through the indicator function $\mathds{1}(|n|=N)$. In that representation, however, the above heuristic is much less transparent.

The following lemma allows us to replace sums over energy levels by integrals. 

\begin{lemma}
\label{lemma-sum converges to integral}
Let $\phi : (0,\infty) \to \mathbb{R}$ be a continuous function. We assume that there are constants $C>0$, $0 \leq p < \alpha$ such that $|\phi(x)| \leq C/x^p$ holds for all $x \in (0,1]$, and that
\begin{equation}
    \label{Condition 2.5}
    \sum_{k=1}^{\infty} m_k (k+1)^{\alpha} < + \infty, \quad \text{ where } \quad m_k = \max_{k \leq x \leq k+1} | \phi(x) |.
\end{equation}
Then,
\begin{equation}
    \lim_{\beta \downarrow 0} \beta^{\alpha} \sum_{j=1}^{\infty} \phi(\beta E_j) = L \alpha \int_0^{\infty} \phi(x) x^{\alpha-1} \mathrm{d}x.
\end{equation}
\end{lemma}
\begin{proof}
    The proof is a straightforward adaptation of the proof of \cite[Lemma~3.1]{ChatDiac2014}. We point out that, although the authors of \cite{ChatDiac2014} state and prove Lemma~3.1 only for continuous functions, they subsequently apply it to functions with a power law singularity at $x=0$. Our Lemma~\ref{lemma-sum converges to integral} resolves this inconsistency. 
\end{proof}

Lemma~\ref{lemma-sum converges to integral} can be used to compute the limit of $\mathbb{E}(M)/N$ as $N\to \infty$. For $b > b_{\mathrm{c}}$ with $b_{\mathrm{c}}$ in \eqref{eq:BECPhaseTransition} we find
\begin{equation}
    \lim_{N\to\infty}\frac{\mathbb{E}(M)}{N} = \left( \frac{b_{\mathrm{c}}}{b} \right)^{\alpha}.
    \label{asymtotic of E(M) for alpha = 1}
\end{equation}
The details of this computation can be found in Lemma~3.2 ($\alpha > 1$) and Lemma~3.3 ($\alpha=1$) in \cite{ChatDiac2014}.

Recall the definitions of $\alpha$ and $L$ in \eqref{eq:energyAsymptotics}. The large deviations behavior of the random variable $Z_0$ in \eqref{definition of M} is captured in the following three propositions, which will be proved in Sections~\ref{sec:left tail}, \ref{sec:right tail}, and \ref{sec:right tail1}.

\begin{proposition}[Left tail for $\alpha \geq 1$]
\label{prop:left tail} 
Let $\alpha \geq 1$, choose $\beta = \beta_N$ as in \eqref{eq:beta} with some $b > b_{\mathrm{c}}$, and consider the limit $N \to \infty$. For any fixed $\eta > 0$,
\begin{equation}
    \lim_{N\to\infty}\frac{1}{N\beta}\log
\mathbb P\left(Z_0<(1-\eta)\mathbb E[Z_0]\right)
= \eta\left[\left(\frac{b_{\mathrm{c}}}{b}\right)^{\alpha}-1 \right]E_1. 
\end{equation}
\end{proposition}

\begin{proposition}[Right tail for $\alpha > 1$]
\label{prop:right tail}
Let $\alpha > 1$, choose $\beta = \beta_N$ as in \eqref{eq:beta} with some $b > b_{\mathrm{c}}$, and consider the limit $N \to \infty$. For any fixed $0 < \eta < [(b/b_{\mathrm{c}})^\alpha - 1]^{-1}$,
\begin{equation}
    \lim_{N\to\infty}\frac{1}{N}\log \mathbb P\left(Z_0>(1+\eta)\mathbb E[Z_0]\right)
= \inf_{s<0}\phi_{\eta,b}(s),    
\end{equation}
where $\phi_{\eta,b}$ is defined in \eqref{right tail rate function}. 
Moreover, $\mathbb P(Z_0 > (1+\eta) \mathbb{E}[Z_0] ) = 0$
if $\eta > [(b/b_{\mathrm{c}})^\alpha - 1]^{-1}$. 
\end{proposition}

In the case $\alpha=1$, we need a slightly more general statement. Its proof is more subtle than those of the previous propositions, as reflected by the fact that we require the limit in \eqref{Euler's Constant} to exist. Note that the existence of this limit does not follow from \eqref{eq:energyAsymptotics}.

\begin{proposition}[Right tail for $\alpha = 1$]
\label{prop:right tail1}
Let $\alpha = 1$, choose $\beta = \beta_N$ as in \eqref{eq:beta} with some $b > b_{\mathrm{c}}$, and consider the limit $N \to \infty$. For any fixed $0 < \eta < [b/b_{\mathrm{c}} - 1]^{-1}$ and any $0 \leq \epsilon_N \leq C N^{-\sigma}$ with some $\sigma,C>0$, we have
\begin{equation}
    \lim_{N\to\infty} \frac{\beta^{\eta [b/b_{\mathrm{c}}-1]}}{[b \log(b/\beta)]^{\eta b/b_\mathrm{c}}} \log \mathbb P\left(Z_0>(1+\eta+\epsilon_N)\mathbb E[Z_0]\right) = -L e^{-(1+\eta) C_{\mathrm{E}}}     
\end{equation}
with $C_{\mathrm{E}}$ as in \eqref{Euler's Constant}. Moreover, $\mathbb P(Z_0 > (1+\eta+\epsilon_N) \mathbb{E}[Z_0] ) = 0$ if $\eta > [b/b_{\mathrm{c}} - 1]^{-1}$.
\end{proposition}

It is no coincidence that $Z_0$ satisfies the same large-deviation bounds as $N_0$ defined in \eqref{eq:beta}. As Lemma~\ref{lemma~conditon} shows,
\begin{equation}
    \mathbb{P}^{\mathrm{can}}(N_0 \in [c,d]) = \mathbb{P}(Z_0 \in [c,d] \ | M \leq N)   
\end{equation}
holds for all $ -\infty \leq c < d \leq +\infty $. Thus, if the preceding propositions are correct, the effect of conditioning on the event ${M\leq N}$ must vanish in the limit. As we show below, this is indeed the case. This observation was already used in \cite{ChatDiac2014} to determine the limiting distributions of $N_0$, where it was shown that $\mathbb{P}(M\leq N)$ converges to $1$ as $N\to\infty$. See also the discussion below Lemma~\ref{lemma~conditon}. Here, however, we are interested in large-deviation estimates for $N_0$, that is, in estimates for probabilities that are exponentially small in $N$ (if $\alpha > 1$). In this setting, one must compare the rate at which $\mathbb{P}(M>N)$ converges to zero with the relevant large-deviation scale.

As shown below, $\mathbb{P}(M>N)$ is much smaller than the right-hand side of \eqref{eq:main1}. This is sufficient to show that conditioning on the event $\{M\leq N \}$ can be treated as a perturbation. Interestingly, $\mathbb{P}(M>N)$ is much larger than the right-hand sides of \eqref{eq:main2} and \eqref{eq:main3}. In these cases, however, the constraint is irrelevant because it is automatically satisfied on the events under consideration. 

We provide more details for the case $\alpha>1$. By \eqref{definition of M}, the condition $Z_0\geq 0$ implies that $M\leq N$. Hence
\begin{equation}
    \frac{1}{N}\log \mathbb{P}^{\mathrm{can}}(N_0 > (1+\eta) \mathbb{E}(N_0) ) = \frac{1}{N} \log \mathbb{P}(Z_0 > (1+\eta) \mathbb{E}(N_0))- \frac{1}{N} \log(1 - P(M > N)).
\end{equation}
As we will show below, the first term converges to the claimed limit as $N\to\infty$, whereas the second term tends to zero at a stretched-exponential rate divided by $N^{-1}$. 

The remainder of the paper is organized as follows. The propositions stated above are proved in Sections~\ref{sec:left tail}, \ref{sec:right tail}, and \ref{sec:right tail1}. In Section~\ref{sec:proofTheorem1}, we use these results to prove Theorem~\ref{thm:main}, following the heuristic argument outlined above.
\subsection{Proof of Proposition~\ref{prop:left tail}}
\label{sec:left tail}
We start by proving an upper bound.
\subsubsection*{Upper bound}
Fix \(\eta>0\). Since \(Z_0=N-M\) and \(\mathbb{E}(Z_0)=N-\mathbb{E}(M)\),
\begin{equation}
    \{Z_0<(1-\eta)\mathbb{E}(Z_0)\} = \{M>\mathbb{E}(M)+\eta(N-\mathbb{E}(M))\},
\end{equation}
and hence an application of Markov's inequality shows that
\begin{equation}
\mathbb{P}\bigl(Z_0<(1-\eta)\mathbb{E}(Z_0)\bigr) \leq e^{-\beta \kappa (\mathbb{E}(M)+\eta(N-\mathbb{E}(M)))} \mathbb{E}(e^{\beta \kappa M})
\label{Markov inequality}
\end{equation}
holds for any $0 < \kappa < E_1$. Using \eqref{eq:defZj}, \eqref{definition of M} and the independence of the family $ \{ Z_j \}_{j=1}^{\infty} $, we compute
\begin{equation}
    \mathbb{E}(e^{\beta \kappa M}) = \prod_{j=1}^\infty \frac{1-e^{-\beta E_j}}{1-e^{\beta(\kappa-E_j)}}.  
    \label{eq:firstProp1}
\end{equation}
In combination, \eqref{Markov inequality} and \eqref{eq:firstProp1} imply
\begin{equation}
    \frac{1}{N \beta} \log \mathbb{P}\bigl(Z_0<(1-\eta)\mathbb{E}(Z_0)\bigr) \leq f_N(\kappa) 
    \label{eq:LDUpperBound}
\end{equation}
with the rate function
\begin{equation}
    f_N(\kappa) = \frac{1}{N\beta}\sum_{j=1}^\infty
\log\left(\frac{1-e^{-\beta E_j}}{1-e^{\beta(\kappa-E_j)}}\right)
- \left[\frac{\mathbb{E}(M)}{N}+\eta \left(1-\frac{\mathbb{E}(M)}{N}\right)\right]\kappa. 
\label{Function f_N for the left tail}
\end{equation}

Let us have a closer look at the first term on the right-hand side of \eqref{Function f_N for the left tail}. We have
\begin{equation}
    \lim_{N \to \infty }\frac{1}{N\beta}\sum_{j=1}^\infty
\log\left(\frac{1-e^{-\beta E_j}}{1-e^{\beta(\kappa-E_j)}}\right) = \lim_{N \to \infty} \int_0^{\kappa} \left[ \frac{1}{N} \sum_{j=1}^{\infty} \frac{1}{e^{\beta (E_j-t)}-1} \right] \mathrm{d}t \to \left( \frac{b_{\mathrm{c}}}{b} \right)^{\alpha} \kappa.
\label{eq:firstProp2Andi4a}
\end{equation}
To compute the limit we used \eqref{asymtotic of E(M) for alpha = 1} for the energy levels $\{ E_j - t \}_{j=1}^{\infty}$ and applied dominated convergence with dominating function $2 (b_{\mathrm{c}}/b)^{\alpha}$. Putting \eqref{asymtotic of E(M) for alpha = 1}, \eqref{eq:LDUpperBound}, \eqref{Function f_N for the left tail}, and \eqref{eq:firstProp2Andi4a} together, we find
\begin{equation}
    \limsup_{N\to\infty}\frac{1}{N\beta}\log \mathbb{P}\bigl(Z_0<(1-\eta)\mathbb{E}(Z_0)\bigr) \leq \inf_{0 < \kappa < E_1}
\eta \kappa \left[ \left(\frac{b_{\mathrm{c}}}{b}\right)^\alpha-1\right] = \eta E_1 \left[ \left(\frac{b_{\mathrm{c}}}{b}\right)^\alpha-1\right].
\label{pointwise convergence of g_N}
\end{equation}
In the last step we used $b > b_{\mathrm{c}}$ and $\eta > 0$. 
\subsubsection*{Lower bound}

We start by estimating the variance of $M$. The independence of the family $\{Z_j\}_{j=1}^{\infty}$ and \eqref{eq:defZj} give
\begin{equation*}
    \mathrm{Var}(M)=\sum_{j=1}^{\infty}\mathrm{Var}(Z_j)=\sum_{j=1}^{\infty}\frac{1}{4\sinh^2(\beta E_j/2)}.
\end{equation*}
We split the sum into two parts, one where $\beta E_j\leq1$ and one where $\beta E_j>1$. An application of Lemma~\ref{lemma-sum converges to integral} shows that the sum over all $j$ with $\beta E_j>1$ is bounded from above by a constant times $\beta^{-\alpha}$. To obtain a bound for the other sum we use $E_j \geq K j^{1/\alpha}$, which follows from \eqref{eq:energyAsymptotics} for some $K > 0$ and all $j \geq 1$, and $\sinh(x) \geq x$ for $x \geq 0$ to estimate
\begin{equation*}
    \sum_{\beta E_j\leq1}\frac{1}{4\sinh^2(\beta E_j/2)}
    \leq\frac{1}{(K\beta)^2}\sum_{1\leq j\leq(K\beta)^{-\alpha}}\frac{1}{j^{2/\alpha}}.
\end{equation*}
The sum on the right-hand side is bounded by a constant if $1\leq\alpha<2$, by a constant times $|\log(\beta)|$ if $\alpha=2$, and by a constant times $\beta^{2-\alpha}$ if $\alpha>2$. In combination, these considerations show that there is a constant $C(\alpha)>0$ such that
\begin{equation}
    \mathrm{Var}(M)\leq C(\alpha)\begin{cases}
        \beta^{-2} & \text{if }1\leq\alpha<2,\\
        \beta^{-2}|\log(\beta)| & \text{if }\alpha=2,\\
        \beta^{-\alpha} & \text{if }\alpha>2.
    \end{cases}
    \label{eq:convOfTildeM7}
\end{equation}
Using \eqref{eq:beta}, one checks that $N^{-2}\mathrm{Var}(M)\to0$ as $N\to\infty$.

We denote $\widetilde{M}=M-Z_1$ and fix $\epsilon>0$. From \eqref{eq:beta} and \eqref{expectation of Z_j} we know that $\mathbb{E}(Z_1)\leq(\beta E_1)^{-1}=o(N)$. Since $\mathbb{E}(M)=\mathbb{E}(Z_1)+\mathbb{E}(\widetilde{M})$, an application of Chebyshev's inequality shows
\begin{align}
    \mathbb{P}\bigl(\widetilde{M}<\mathbb{E}(M)-\epsilon N\bigr)
    &=\mathbb{P}\bigl(\mathbb{E}(\widetilde{M})-\widetilde{M}>\epsilon N-\mathbb{E}(Z_1)\bigr) \nonumber \\
    &\leq\mathbb{P}\bigl(|\widetilde{M}-\mathbb{E}(\widetilde{M})|>\epsilon N/2\bigr)
    \leq\frac{4\mathrm{Var}(\widetilde{M})}{\epsilon^2N^2}
    \leq\frac{4\mathrm{Var}(M)}{\epsilon^2N^2},
\end{align}
provided $N$ is large enough. By our previous consideration, the right-hand side goes to zero for $N \to \infty$. In particular,
\begin{equation}
    \lim_{N\to\infty}\mathbb{P}\bigl(\widetilde{M}\geq\mathbb{E}(M)-\epsilon N\bigr)=1.
    \label{eq:leftTailProbabilityS}
\end{equation}

The independence of $Z_1$ and $\widetilde{M}$ allow us to estimate
\begin{align}
    \mathbb{P}\bigl(Z_0<(1-\eta)\mathbb{E}(Z_0)\bigr)
    &=\mathbb{P}\bigl(M>\mathbb{E}(M)+\eta(N-\mathbb{E}(M))\bigr)\nonumber\\
    &\geq\mathbb{P}\bigl(\widetilde{M}\geq\mathbb{E}(M)-\epsilon N\bigr)
    \mathbb{P}\bigl(Z_1>\eta(N-\mathbb{E}(M))+\epsilon N\bigr)\nonumber\\
    &\geq\mathbb{P}\bigl(\widetilde{M}\geq\mathbb{E}(M)-\epsilon N\bigr)
    e^{-\beta E_1[\eta(N-\mathbb{E}(M))+\epsilon N+1]}.
    \label{eq:leftTailLowerBound}
\end{align}
The last inequality follows from $\mathbb{P}(Z_1>t)=e^{-\beta E_1(\lfloor t\rfloor+1)}\geq e^{-\beta E_1(t+1)}$ for $t\geq0$, and the fact that $N - \mathbb{E}(M) > 0$ for all sufficiently large $N$, since $b>b_{\mathrm{c}}$.

Taking logarithms on both sides of \eqref{eq:leftTailLowerBound}, dividing by $N\beta$, and using \eqref{asymtotic of E(M) for alpha = 1}, \eqref{eq:leftTailProbabilityS}, and $N\beta\to\infty$, we find
\begin{equation}
    \liminf_{N\to\infty}\frac{1}{N\beta}\log\mathbb{P}\bigl(Z_0<(1-\eta)\mathbb{E}(Z_0)\bigr)
    \geq\eta E_1\left[\left(\frac{b_{\mathrm{c}}}{b}\right)^{\alpha}-1\right]-\epsilon E_1.
    \label{eq:convOfTildeM12}
\end{equation}
We take the limit $\epsilon\to0$ on both sides. Together with \eqref{pointwise convergence of g_N}, this proves Proposition~\ref{prop:left tail}.
\subsection{Proof of Proposition \ref{prop:right tail}}
\label{sec:right tail}
We start again with the proof of an upper bound. 
\subsubsection*{Upper bound}
Let us assume that $b > b_{\mathrm{c}}$ and $0 < \eta \leq [(b/b_{\mathrm{c}})^\alpha - 1]^{-1}$. We argue as in the analysis that leads to \eqref{eq:LDUpperBound}, to see that
\begin{equation}
    \frac{1}{N} \log \mathbb{P}\bigl(Z_0>(1+\eta)\mathbb{E}(Z_0)\bigr) \leq I_{N}(s) 
    \label{eq:Prop341}
\end{equation}
holds for any $s<0$, where 
\begin{equation}
    I_{N}(s) = \frac{1}{N}\sum_{j=1}^\infty
\log\left(\frac{1-e^{-\beta E_j}}{1-e^{s-\beta E_j}}\right)
- \left[\frac{\mathbb{E}(M)}{N}-\eta \left(1-\frac{\mathbb{E}(M)}{N}\right)\right] s. 
\label{Function f_N for the right tail}
\end{equation}
From \eqref{right tail rate function}, Lemma~\ref{lemma-sum converges to integral}, and \eqref{asymtotic of E(M) for alpha = 1} we know that
\begin{equation}
    \lim_{N \to \infty} I_{N}(s) = L\alpha b^{-\alpha} \int_0^\infty
\log\left(\frac{1-e^{-x}}{1-e^{s-x}}\right) x^{\alpha-1} \mathrm{d}x - \Bigg[ \left(\frac{b_{\mathrm{c}}}{b}\right)^\alpha
- \eta \left(1-\left(\frac{b_{\mathrm{c}}}{b}\right)^\alpha\right) \Bigg]s = \phi_{\eta,b}(s),
\label{eq:Prop342}
\end{equation}
and hence
\begin{equation}
    \limsup_{N \to \infty} \frac{1}{N} \log \mathbb{P}\bigl(Z_0>(1+\eta)\mathbb{E}(Z_0)\bigr) \leq \inf_{s<0}\phi_{\eta,b}(s).
    \label{eq:finalEqUpperBoundProp2.4}
\end{equation}
\subsubsection*{Lower bound}
For $s < 0$ we introduce the independent family of random variables $\{ Y_j \}_{j=1}^{\infty}$ with laws
\begin{equation}
    \mathbb{P}(Y_j=k) = \bigl(1-e^{s-\beta E_j}\bigr) e^{(s-\beta E_j)k}, \qquad k=0,1,2,\dots    
    \label{eq:defYj}
\end{equation}
and define $Y(s) = \sum_{j=1}^\infty Y_j$. In the following lemma we show that $\beta^{\alpha}Y$ converges in probability to $\lim_{\beta \downarrow 0} \beta^{\alpha} \mathbb{E}[Y(s)]$, which is explicitly computable, as $\beta \to 0$.
\begin{lemma}
\label{lem:limitingProbabilityProp24}
   Assume that the hypothesis of Proposition~\ref{prop:right tail} hold and fix $s<0$. Then
        \begin{equation}
        \beta^{\alpha} Y(s) \xrightarrow{\mathbb{P}} L \alpha \int_0^{\infty} \frac{1}{e^{x-s}-1} x^{\alpha-1} \mathrm{d}x
    \end{equation}
    as $\beta \downarrow 0$, where $\xrightarrow{\mathbb{P}}$ stands for convergence in probability.
\end{lemma}
\begin{proof}
    An application of Lemma~\ref{lemma-sum converges to integral} shows 
    \begin{equation}
        \lim_{\beta \downarrow 0} \beta^{\alpha} \mathbb{E}(Y) = \lim_{\beta \downarrow 0} \beta^{\alpha} \sum_{j=1}^{\infty} \frac{1}{e^{\beta E_j - s}-1} = L \alpha \int_0^{\infty} \frac{1}{e^{x-s}-1} x^{\alpha-1} \mathrm{d}x.
    \end{equation}
    A straightforward adaption of its proof shows that \eqref{eq:convOfTildeM7} also holds for the variance of $Y(s)$. In combination with an application of Markov's inequality, this proves the claim of the lemma.
\end{proof}

Let us first consider the case $0 < \eta < [(b/b_{\mathrm{c}})^\alpha - 1]^{-1}$. For such values of $\eta$ the function $\phi_{\eta,b}(s)$, $s<0$ in \eqref{eq:Prop342} has the following properties:
\begin{equation}
    \lim_{s \to -\infty} \phi_{\eta,b}(s) = +\infty, \qquad \lim_{s \to 0} \phi_{\eta,b}(s) = 0,
    \label{eq:lowerBoundProf241}
\end{equation}
and its one-sided derivative at $s = 0$ is positive. Moreover, it is strictly convex. We conclude that it has a unique global minimum $s_0$ in the interval \((-\infty,0)\). The minimizer \(s_0\) satisfies the equation
\begin{equation}
    L\alpha b^{-\alpha} \int_0^\infty
\frac{1}{e^{x-s_0} -1 } x^{\alpha-1} \mathrm{d}x = \left(\frac{b_{\mathrm{c}}}{b}\right)^\alpha
- \eta \left(1-\left(\frac{b_{\mathrm{c}}}{b}\right)^\alpha\right). 
\label{eq:lowerBoundProf242}
\end{equation}
We note that the left-hand side of this equation also equals $\lim_{N \to \infty} \mathbb{E}[Y(s_0)]/N$. 
Together with \eqref{asymtotic of E(M) for alpha = 1}, this allow us to conclude that
\begin{equation}
    \lim_{N \to \infty} \left[ \frac{\mathbb{E}[Y(s_0)]}{N} - \frac{\mathbb{E}(M)}{N}
+ \eta \left(1-\frac{\mathbb{E}(M)}{N}\right) \right] = 0.
\label{eq:lowerBoundProf243}
\end{equation}
In the following we use the notations 
\begin{equation}
    a_N = \frac{\mathbb{E}(M) - \eta(N - \mathbb{E}(M))}{N} \quad \text{ and } \quad m(s) = \lim_{N \to \infty} \frac{\mathbb{E}[Y(s)]}{N}.
    \label{eq:definitionaN}
\end{equation}
We note that the map $s \mapsto m(s)$ is strictly monotone increasing.

The law $\mu_{Z_j}$ of $Z_j$ is absolutely continuous with respect to the law $\mu_{Y_j}$ of $Y_j$ with Radon--Nikodym derivative  
\begin{equation}
    \frac{\mathrm{d}\mu_{Z_j}}{\mathrm{d}\mu_{Y_j}}(x)=\frac{1-e^{-\beta E_j}}{1-e^{s-\beta E_j}}\,e^{-s x}.    
    \label{eq:RadonNykodym}
\end{equation}
We also note that
\begin{equation}
    \prod_{j=1}^{\infty}\frac{1-e^{-\beta E_j}}{1-e^{s-\beta E_j}}=e^{NI_N(s)+sNa_N},
\end{equation}
where $I_N$ and $a_N$ are defined in \eqref{Function f_N for the right tail} and \eqref{eq:definitionaN}, respectively. For the characteristic function of $M$ this implies
\begin{align}
\mathbb{E}\left(e^{i\xi M}\right)&=\prod_{j=1}^{\infty}\mathbb{E}\left(e^{i\xi Z_j}\right)=\prod_{j=1}^{\infty}\int_{\mathbb{R}} e^{i\xi x}\, \mathrm{d}\mu_{Z_j}(x)=\prod_{j=1}^{\infty}\int_{\mathbb{R}}e^{i\xi x}\frac{1-e^{-\beta E_j}}{1-e^{s-\beta E_j}}e^{-s x}\,\mathrm{d}\mu_{Y_j}(x) \nonumber \\
&=\prod_{j=1}^{\infty}\mathbb{E}\left(e^{(i\xi-s)Y_j}\frac{1-e^{-\beta E_j}}{1-e^{s-\beta E_j}}\right)=\mathbb{E}\left( e^{(i\xi-s)Y(s)} e^{N I_N(s)+sNa_N}\right),
\label{eq:convOfTildeM10}
\end{align}
and hence
\begin{equation}
    \mathbb{P}(M \in B) = e^{N I_N(s)} \mathbb{E}[e^{-s(Y(s) - N a_N)} \mathbf{1}_{B}(Y(s)) ]
    \label{eq:proofProp243}
\end{equation}
holds for any Borel set $B \subseteq \mathbb{R}$ and for any $s < 0$. Here $\mathbf{1}_B$ denotes the indicator function of $B$. 

Pick some $s < s_0$. Since $m(s) < m(s_0)$ we can also pick $\epsilon > 0$ such that $\epsilon < m(s_0)-m(s)$. Finally, we pick $N'>0$ such that 
\begin{equation}
    \{ m(s) -\epsilon < Y(s)/N < m(s) + \epsilon \} \subseteq \{ Y(s)/N < a_N \}
\end{equation}
holds for all $N > N'$. This is possible because $\lim_{N \to \infty} a_N = \lim_{N \to \infty} \mathbb{E}[Y(s_0)]/N = m(s_0) > m(s)+\epsilon$ by \eqref{eq:lowerBoundProf243}. In the following we assume $N > N'$. Using \eqref{eq:proofProp243}, we compute
\begin{align}
    \mathbb{P}(Z_0 > (1+\eta) \mathbb{E}(Z_0)) &= \mathbb{P}(M < \mathbb{E}(M) - \eta (N - \mathbb{E}(M))) = e^{N I_N(s)} \mathbb{E}[ e^{-s(Y(s) - N a_N )} \mathbf{1}_{\{Y(s) < N a_N\}} ] \nonumber \\
    &\geq e^{N I_N(s)} \mathbb{E}\left[ e^{-s(Y(s) - N a_N )} \mathbf{1}_{\{m(s) -\epsilon < Y(s)/N < m(s) + \epsilon\}} \right] \nonumber \\
    &\geq e^{N I_N(s)} e^{-s N (m(s) - \epsilon - a_N)} \mathbb{P}( m(s) -\epsilon < Y(s)/N < m(s) + \epsilon ).
    \label{eq:AlgebraLowerBoundRightTail}
\end{align}
An application of Lemma~\ref{lem:limitingProbabilityProp24} shows $\lim_{N \to \infty} \mathbb{P}( m(s) -\epsilon < Y(s)/N < m(s) + \epsilon ) = 1$. Hence, we have
\begin{align}
    \liminf_{N \to \infty} \frac{1}{N} \log \mathbb{P}(Z_0 > (1+\eta) \mathbb{E}(Z_0)) &\geq \liminf_{N \to \infty} \left[ I_N(s) -s (m(s) - \epsilon - a_N) \right] \nonumber \\
    &= \phi_{\eta,b}(s) - s(m(s) - m(s_0) - \epsilon).
    \label{eq:proofProp244}
\end{align}
The bound holds for all $ 0 < \epsilon < m(s_0) - m(s)$ and for all $s < s_0$. Moreover, the right-hand side is continuous in $s$. Accordingly, we can take the limits $\epsilon \to 0$ and $s \uparrow s_0$ on both sides and obtain
\begin{align}
    \liminf_{N \to \infty} \frac{1}{N} \log \mathbb{P}(Z_0 > (1+\eta) \mathbb{E}(Z_0)) \geq \phi_{\eta,b}(s_0) = \inf_{s < 0} \phi_{\eta,b}(s),
    \label{eq:proofProp245}
\end{align}
which proves the lower bound. In combination with \eqref{eq:finalEqUpperBoundProp2.4}, \eqref{eq:proofProp245} proves Proposition~\ref{prop:right tail} in the case $0 < \eta < [(b/b_{\mathrm{c}})^\alpha - 1]^{-1}$. The proof in the case $\eta > [(b/b_{\mathrm{c}})^\alpha - 1]^{-1}$ is straightforward, and therefore left to the reader.
\subsection{Proof of Proposition~\ref{prop:right tail1}}
\label{sec:right tail1}
\subsubsection*{Upper bound}
Let us assume that $b > b_{\mathrm{c}}$, $0 < \eta < [b/b_{\mathrm{c}} - 1]^{-1}$ and choose a sequence $0 \leq \epsilon_N \leq C N^{-\sigma}$ with some $C,\sigma > 0$. We also introduce the notation $c = \eta [b/b_{\mathrm{c}}-1]$ and note that our assumptions imply $0 < c < 1$. Similarly to \eqref{eq:Prop341}, we have
\begin{equation}
    \frac{\beta^{c}}{[b \log(b/\beta)]^{\eta b/b_\mathrm{c}}} \log \mathbb{P}(Z_0>(1+\eta+\epsilon_N)\mathbb{E}(Z_0)) \leq g_N(s)
    \label{eq:proofProposition251}
\end{equation}
for any $s<0$ with the rate function 
\begin{align}
    g_N(s) &= \frac{\beta^{c}}{[b \log(b/\beta)]^{\eta b/b_\mathrm{c}}} \left\{ \sum_{j=1}^\infty
\log\left( \frac{1-e^{-\beta E_j}}{1-e^{s-\beta E_j}} \right)
- \left[ \mathbb{E}(M)-\eta \left( N-\mathbb{E}(M) \right) \right] s \right\} + A_N(s), \quad \text{where} \nonumber \\
    A_N &= \frac{\beta^{c}}{[b \log(b/\beta)]^{\eta b/b_\mathrm{c}}} \epsilon_N \left( N-\mathbb{E}(M) \right) s.
\label{eq:proofProposition252}
\end{align}
In contrast to the proof of Proposition~\ref{prop:right tail}, we cannot take the limit $N \to \infty$ for $g_N(s)$ for fixed $s<0$ and then optimize over $s$ because $\lim_{N \to \infty} g_N(s) = + \infty$. As we will see below, $\lim_{N \to \infty} \inf_{s < 0} g_N(s)$ equals the asymptotics claimed in Proposition~\ref{prop:right tail1}, and hence this is the quantity we need to compute.

The function $g_N(s)$, $s \leq 0$ satisfies $g_N(0) = 0$, its one sided derivative at $s = 0$ is strictly positive and $\lim_{s \to -\infty} g_N(s) = + \infty$ holds for $N$ large enough. Moreover, $g_N$ is strictly convex. For $N$ large enough, it therefore has a unique minimizer $s_N < 0$ that solves the equation
\begin{equation}
    \sum_{j=1}^{\infty} \frac{1}{e^{\beta E_j - s_N} - 1} = \mathbb{E}(M)-(\eta + \epsilon_N) \left( N-\mathbb{E}(M) \right).
    \label{eq:proofProposition253}
\end{equation}
In the following lemma we compute the asymptotic behavior of $s_N$.

\begin{lemma}
\label{lem:asymptoticsOfsN}
    Assume that the hypotheses of Proposition~\ref{prop:right tail1} hold. Then the unique minimizer $s_N$ of $g_N$ in \eqref{eq:proofProposition252} satisfies
    \begin{equation}
        s_N = - \beta^{1-c} [b \log(b/\beta) ]^{\eta b/b_{\mathrm{c}}} e^{-(1+\eta) C_{\mathrm{E}} + o(1) }
        \label{eq:proofProposition254}
    \end{equation}
    in the limit $N \to \infty$, where $C_{\mathrm{E}}$ has been defined in \eqref{Euler's Constant}.
\end{lemma}
\begin{proof}
    From \eqref{Euler's Constant} and \eqref{asymtotic of E(M) for alpha = 1} we know that
    \begin{equation}
        \mathbb{E}(M) = \frac{L}{\beta} \left[ |\log(\beta)| + C_{\mathrm{E}} + o(1) \right].
        \label{eq:proofProposition255}
    \end{equation}
    Moreover, \eqref{eq:beta} implies
    \begin{equation}
        N = \frac{b}{\beta} \left[ \log\left( \frac{b}{\beta} \right) + \log \log\left( \frac{b}{\beta} \right) + o(1) \right].
        \label{eq:proofProposition256}
    \end{equation}
    Putting \eqref{eq:proofProposition255} and \eqref{eq:proofProposition256} together, we find
    \begin{align}
        \mathbb{E}(M)-\eta \left( N-\mathbb{E}(M) \right) =& \frac{(1+\eta) L}{\beta} \left[ |\log(\beta)| + C_{\mathrm{E}} + o(1) \right] \nonumber \\
        &- \eta \frac{b}{\beta} \left[ \log\left( b/\beta \right) + \log \log\left( b/\beta \right) + o(1) \right]. 
        \label{eq:proofProposition257}
    \end{align}
    Moreover,
    \begin{equation}
        0 \leq \epsilon_N ( N - \mathbb{E}(M) ) \leq \epsilon_N N \leq C N^{1-\sigma}.
        \label{eq:proofProposition258}
    \end{equation}

    Let's now consider the term on the left-hand side of \eqref{eq:proofProposition253}. We will assume that $s_N \to 0$ and $\beta/ s_N \to 0$ as $N \to \infty$, which can be justified as follows. If $|s_N|$ is bounded away from $0$ uniformly in $N$ then the left-hand side of \eqref{eq:proofProposition253} can be approximated by a finite integral, while the same integral representation of the right-hand side would be singular. One can easily check that this yields a contradiction. If $\beta/s_N$ does not go to $0$ then $N^{-1}$ times the left-hand side of \eqref{eq:proofProposition253} converges to $b_{\mathrm{c}}/b$, which is different from the limit of $N^{-1}$ times the right-hand side unless $\eta = 0 = \epsilon_N$. 
    
    We claim that
    \begin{equation}
        \frac{\beta}{L} \left[ \sum_{j=1}^{\infty} \frac{1}{e^{\beta E_j} - 1}  - \sum_{j=1}^{\infty} \frac{1}{e^{\beta E_j - s_N} - 1} \right] = \log(|s_N|/\beta) + C_{\mathrm{E}} + o(1).
        \label{eq:proofProposition258}
    \end{equation}
    To prove this, we first will show that
    \begin{equation}
        \frac{\beta}{s_N L} \sum_{j=1}^{\infty} \left[ \frac{|s_N|}{e^{\beta E_j} - 1} - \frac{|s_N|}{e^{\beta E_j - s_N} -1} - \frac{1}{e^{\beta E_j/|s_N|}-1} \right] \to 0
        \label{eq:proofProposition259}
    \end{equation}
    for $N \to \infty$. 
    
    Let us introduce the notation $t_N = |s_N|$, $\delta_N = \beta/t_N$, which satisfy $t_N, \delta_N \to 0$ as $N \to \infty$. We also introduce the functions
    \begin{equation}
        h_t(x) = \frac{t}{e^{tx}-1} - \frac{t}{e^{t(x+1)}-1} - \frac{1}{e^x-1} \quad \text{ and } \quad h_0(x) = \frac{1}{x} - \frac{1}{x+1} - \frac{1}{e^x-1}.
        \label{eq:proofProposition259a}
    \end{equation}
    We have 
    \begin{equation}
        \frac{t}{e^{tx}-1} = \frac{1}{x} - \frac{t}{2} +O(t^2 x),
        \label{eq:proofProposition259aa}
    \end{equation}
    which holds uniformly in $x \in (0,x_0]$ for small $t > 0$ and any fixed $x_0>0$. A short computation that uses this expansion shows that $h_t(x)$ is uniformly bounded for $x \in (0,x_0]$ and $0 < t \leq 1$. Note also that $h_0(x)$ is bounded for $x \in (0,1]$. The expression on the left-hand side of \eqref{eq:proofProposition259} equals $-1/L$ times $\delta_N \sum_{j=1}^{\infty} h_{t_N}(\delta_N E_j)$. 
    
    For fixed $R > 0$ we choose a smooth function $\chi_R : [0,\infty) \to [0,1]$, which satisfies $\chi_R(x) = 1$ for $x \in [0,R]$, and $\chi_R(x) = 0$ for $x \in [2R,\infty)$, and write
    \begin{align}
        \delta_N \sum_{j=1}^{\infty} h_{t_N}(\delta_N E_j) =& \delta_N \sum_{j=1}^{\infty} \chi_R(\delta_N E_j) h_0(\delta_N E_j) + \delta_N \sum_{j=1}^{\infty} \chi_R(\delta_N E_j) \left[ h_{t_N}(\delta_N E_j) - h_0(\delta_N E_j) \right] \nonumber \\
        &+ \delta_N \sum_{j=1}^{\infty} \left[ 1 - \chi_R(\delta_N E_j) \right] h_{t_N}(\delta_N E_j).
        \label{eq:proofProposition259b}
    \end{align}

    An application of Lemma~\ref{lemma-sum converges to integral} shows 
    \begin{equation}
        \lim_{N \to \infty} \delta_N \sum_{j=1}^{\infty} \chi_R(\delta_N E_j) h_0(\delta_N E_j) = L \int_0^{\infty} \chi_R(x) h_0(x) \mathrm{d}x.
        \label{eq:proofProposition259c}
    \end{equation}
    To obtain a bound for the second term on the right-hand side of \eqref{eq:proofProposition259b}, we first use \eqref{eq:proofProposition259aa} to check that $\sup_{x \in (0,2R]} | h_t(x) - h_0(x) | \leq C R t^2$. But this implies
    \begin{equation}
        \left| \delta_N \sum_{j=1}^{\infty} \chi_R(\delta_N E_j) \left[ h_{t_N}(\delta_N E_j) - h_0(\delta_N E_j) \right] \right| \leq C R \delta_N t_N^2 \sum_{j=1}^{\infty} \chi_R(\delta_N E_j) \leq C R \delta_N t_N^2 N(2R/\delta_N)
        \label{eq:proofProposition259d}
    \end{equation}
    with 
    \begin{equation}
        N(\lambda) = \#\{j:\,E_j\le \lambda\}.
        \label{eq:CountingFunction}
    \end{equation}
    Our assumption on $N(\lambda)$ in \eqref{eq:energyAsymptotics} guarantees that the right-hand side is bounded by a constant times $t_N^2 R^2$. 

    Next, we consider the third term on the right-hand side of \eqref{eq:proofProposition259b}. We use the bound $\sinh(x) \geq x$ for $x \geq 0$ to see that 
    \begin{equation}
        \left| \frac{t}{e^{tx}-1} - \frac{t}{e^{t(x+1)}} \right| = \left| \int_0^1 \frac{t^2}{4 \sinh^2\left( \frac{t(x + s)}{2} \right)}  \mathrm{d}s \right| \leq \frac{1}{x^2}
        \label{eq:proofProposition2510}
    \end{equation}
    holds for $t \in (0,1]$ and $x > 0$. In particular, $|h_t(x)| \leq C/x^2$ holds with some $C>0$ for such $t$ and $x$. Using this, we estimate
    \begin{align}
        \left| \delta_N \sum_{j=1}^{\infty} \left[ 1 - \chi_R(\delta_N E_j) \right] h_{t_N}(\delta_N E_j) \right| &\leq C \delta_N^{-1} \sum_{j=1}^{\infty} \left[ 1 - \chi_R(\delta_N E_j) \right] \frac{1}{E_j^2} \nonumber \\
        &\leq C K^2 \delta_N^{-1} \sum_{j \geq 2 R/(K' \delta_N) } \frac{1}{j^2}.
        \label{eq:proofProposition259e}
    \end{align}
    To obtain the last bound, we used that $K j \leq  E_j \leq K' j$ holds for some $K, K' > 0$. The right-hand side is bounded by a constant times $R^{-1}$.

    When we put \eqref{eq:proofProposition259b}--\eqref{eq:proofProposition259d}, and \eqref{eq:proofProposition259e} together, first take $\liminf_{N\to \infty}/\limsup_{N \to \infty}$ followed by the limit $R \to \infty$ on both sides of \eqref{eq:proofProposition259b}, we find
    \begin{equation}
        \lim_{N \to \infty} \delta_N \sum_{j=1}^{\infty} h_{t_N}(\delta_N E_j) = L \int_0^{\infty} h_0(x) \mathrm{d}x.
        \label{eq:proofProposition259f}
    \end{equation}
    The integrand on the right-hand side can be written as 
    \begin{equation}
        \frac{1}{x} - \frac{1}{x+1} - \frac{1}{e^x - 1} = \frac{\mathrm{d}}{\mathrm{d}x} \left[ \log\left( \frac{x}{x+1} \right) - \log\left( 1 - e^{-x} \right) \right],
        \label{eq:proofProposition2512}
    \end{equation}
    which allows us to conclude that the integral on the right-hand side of \eqref{eq:proofProposition259f} equals zero. Hence, \eqref{eq:proofProposition259} holds.

    The proof of \eqref{eq:proofProposition258} follows from \eqref{eq:proofProposition259} and 
    \begin{equation}
        \frac{\beta}{|s_N| L} \sum_{j=1}^{\infty} \frac{1}{e^{\beta E_j/|s_N|}-1} = \log(|s_N|/\beta) + C_{\mathrm{E}} + o(1),
        \label{eq:proofProposition2513}
    \end{equation}
    which is a direct consequence of \eqref{eq:proofProposition255}.

    When we put \eqref{eq:proofProposition253}, \eqref{eq:proofProposition255}, \eqref{eq:proofProposition257} and \eqref{eq:proofProposition258} together, we find that $s_N$ solves the equation
    \begin{equation}
        |\log(|s_N|)| = (1+\eta) \left[ \log(1/\beta) + C_{\mathrm{E}} + o(1) \right] -\eta \frac{b}{b_{\mathrm{c}}} \left[ \log\left( b/\beta \right) + \log \log\left( b/\beta \right) + o(1) \right].
        \label{eq:proofProposition2514}
    \end{equation}
    An application of the exponential function on both sides gives \eqref{eq:proofProposition254}. This proves the claim of the lemma.
\end{proof}

We will now use Lemma~\ref{lem:asymptoticsOfsN} to compute the asymptotics of $g_N(s_N)$ as $N \to \infty$. To that end, we use the Taylor expansion of $\log(1-z)$ around $z=0$ and \eqref{eq:proofProposition253} to write
\begin{equation}
    g_N(s_N) = \frac{\beta^{c}}{[b \log(b/\beta)]^{\eta b/b_\mathrm{c}}} \sum_{k=1}^{\infty} \left[ \frac{e^{s_N k} - 1 - ks_N e^{s_N k}}{k} \right] \left[  \sum_{j=1}^{\infty} e^{-\beta E_j k} \right] + A_N(s_N).
    \label{eq:proofProposition2515}
\end{equation}
Applications of Lemma~\ref{lem:asymptoticsOfsN} and the bound $0 \leq \epsilon_N \leq C N^{-\sigma}$ show
\begin{equation}
    \lim_{N \to \infty} A_N(s_N) = \frac{\beta^{c}}{[b \log(b/\beta)]^{\eta b/b_\mathrm{c}}} \epsilon_N \left( N-\mathbb{E}(M) \right) s_N = 0.
    \label{eq:proofProposition2515b}
\end{equation}
From Lemma~\ref{lemma-sum converges to integral} we know that 
\begin{equation}
    \lim_{\beta k \downarrow 0} \beta k \sum_{j=1}^{\infty} e^{-\beta E_j k} = L \int_0^{\infty} e^{-x} \mathrm{d}x = L.
    \label{eq:proofProposition2516}
\end{equation}
Let us introduce the notation $R(\beta k) = \beta k \sum_{j=1}^{\infty} e^{-\beta E_j k} - L$ and write
\begin{align}
    g_N(s_N) =& \frac{\beta^{c-1} L}{[b \log(b/\beta)]^{\eta b/b_\mathrm{c}}} \sum_{k=1}^{\infty} \left[ \frac{e^{s_N k} - 1 - ks_N e^{s_N k}}{k^2} \right] \nonumber \\
    &\hspace{1cm}+ \frac{\beta^{c-1}}{[b \log(b/\beta)]^{\eta b/b_\mathrm{c}}} \sum_{k=1}^{\infty} \left[ \frac{e^{s_N k} - 1 - ks_N e^{s_N k}}{k^2} \right] R(\beta k).
    \label{eq:proofProposition2517}
\end{align}

For $t<0$ we have 
\begin{equation}
    \frac{\mathrm{d}}{\mathrm{d}t} \sum_{k=1}^{\infty} \left[ \frac{e^{t k} - 1 - k t e^{t k}}{k^2} \right] = - t \sum_{k=1}^{\infty}  e^{t k} = \frac{-t}{ e^{-t} - 1 },
    \label{eq:proofProposition2518}
\end{equation}
which, with an application of Lemma~\ref{lem:asymptoticsOfsN}, allows us to show that
\begin{align}
    \lim_{N \to \infty} \frac{\beta^{c-1} L}{[b \log(b/\beta)]^{\eta b/b_\mathrm{c}}} \sum_{k=1}^{\infty} \left[ \frac{e^{s_N k} - 1 - ks_N e^{s_N k}}{k^2} \right] &= \lim_{N \to \infty}\frac{\beta^{c-1} L}{[b \log(b/\beta)]^{\eta b/b_\mathrm{c}}} \int_{0}^{s_N} \frac{-t}{e^{-t}-1} \mathrm{d}t \nonumber \\
    &= -L e^{-(1+\eta) C_{\mathrm{E}}}.
    \label{eq:proofProposition2519}
\end{align}

Let us denote $\delta_{\beta} = 1/[\beta \log(1/\beta)]$. We split the sum in the second term on the right-hand side of \eqref{eq:proofProposition2517} into two parts, one where $k \leq \delta_{\beta}$ and one with $k > \delta_{\beta}$. We use the bound $|e^{-x} - 1 + x e^{-x}| \leq \min\{x^2/2,1\}$, which holds for $x \geq 0$, to show that the first of these two terms satisfies
\begin{align}
    \left| \sum_{k \leq \delta_{\beta}} \left[ \frac{e^{s_N k} - 1 - ks_N e^{s_N k}}{k^2} \right] R(\beta k)\right| &\leq \left( \sup_{k \leq \delta_{\beta}} |R(\beta k)| \right) \left[ \sum_{1 \leq k \leq t_N^{-1}} \frac{t_N^2}{2} + \sum_{k > t_N^{-1}} \frac{1}{k^2} \right] \nonumber \\
    &\leq C t_N \left( \sup_{k \leq \delta_{\beta}} |R(\beta k)| \right).
    \label{eq:proofProposition2520}
\end{align}
The right-hand side multiplied by $\beta^{c-1} L/[\log(b/\beta)]^{\eta b/b_\mathrm{c}}$ goes to zero for $N \to \infty$ because, by Lemma~\ref{lem:asymptoticsOfsN}, $t_N$ cancels the prefactor and $\beta \delta_{\beta} \to 0$ for $N \to \infty$, which implies 
\begin{equation}
    \lim_{N \to \infty} \sup_{k \leq \delta_{\beta}} |R(\beta k)| = 0.
    \label{eq:proofProposition2520b}
\end{equation}

For $k > \delta_{\beta}$ we have $|(e^{s_N k} - 1 - ks_N e^{s_N k}) R(\beta k)| \leq C$ with some $C>0$. This allows us to estimate
\begin{equation}
     \left| \sum_{k > \delta_{\beta}} \left[ \frac{e^{s_N k} - 1 - ks_N e^{s_N k}}{k^2} \right] R(\beta k)\right| \leq C \sum_{k>\delta_{\beta}} \frac{1}{k^2} \leq C \sigma^{-1}_{\beta}.
    \label{eq:proofProposition2521}
\end{equation}
We put \eqref{eq:proofProposition2520} and \eqref{eq:proofProposition2521} together, apply Lemma~\ref{lem:asymptoticsOfsN} and use $0 < c < 1$ to show that 
\begin{equation}
    \limsup_{N \to \infty} \frac{\beta^{c-1} L}{[b \log(b/\beta)]^{\eta b/b_\mathrm{c}}} \left| \sum_{k = 1}^{\infty} \left[ \frac{e^{s_N k} - 1 - ks_N e^{s_N k}}{k^2} \right] R(\beta k)\right| = 0.
    \label{eq:proofProposition2522}
\end{equation}

In combination, \eqref{eq:proofProposition2515b}, \eqref{eq:proofProposition2517}, \eqref{eq:proofProposition2519} and \eqref{eq:proofProposition2522} give
\begin{equation}
    \lim_{N \to \infty} g_N(s_N) = -L e^{-(1+\eta) C_{\mathrm{E}}}.
    \label{eq:proofProposition2523}
\end{equation}
Finally, together with \eqref{eq:proofProposition251}, \eqref{eq:proofProposition2523} proves
\begin{equation}
    \limsup_{N \to \infty} \frac{\beta^{c}}{[b \log(b/\beta)]^{\eta b/b_\mathrm{c}}} \log \mathbb{P}(Z_0>(1+\eta+\epsilon_N)\mathbb{E}(Z_0)) \leq -L e^{-(1+\eta) C_{\mathrm{E}}}. 
    \label{eq:proofProposition2524}
\end{equation}
It remains to derive a matching lower bound.
\subsubsection*{Lower bound}
We recall the definition of the independent family of random variables $\{ Y_j \}_{j=1}^{\infty}$ in \eqref{eq:defYj} and that of $Y(s)$ below \eqref{eq:defYj}. In the following we choose $s = s_N$.

\begin{lemma}
\label{lem:CLT}
    Assume that the hypotheses of Proposition~\ref{prop:right tail1} hold. Then
    \begin{equation}
        \frac{Y(s_N) - \mathbb{E}[Y(s_N)]}{\sqrt{\mathrm{Var}[Y(s_N)]}} \xrightarrow{d} \mathcal{N}(0,1)
        \label{eq:proofProposition2525}
    \end{equation}
    as $N \to \infty$, where $s_N$ denotes the unique minimizer of $g_N$ in \eqref{eq:proofProposition252} and $\mathrm{Var}[Y(s_N)]$ is the variance of $Y(s_N)$.
\end{lemma}
\begin{proof}
    We will show that the claimed result follows from an application of the version of the Lyapunov central limit theorem in \cite[Theorem~8.4.1]{Borovkov2013}.

    Let us start by computing the asymptotics of the variance of $Y(s_N)$, which can be written as
    \begin{equation}
        \sum_{j=1}^{\infty} \mathrm{Var}(Y_j) = \sum_{k=1}^{\infty} k e^{s_N k} \left( \sum_{j=1}^{\infty} e^{-\beta E_j k} \right) = \frac{1}{\beta} \sum_{k=1}^{\infty} e^{s_N k} \left( L + R(\beta k) \right).
        \label{eq:proofProposition2526}
    \end{equation}
    Here $R(\beta k)$ denotes the function that we introduced below \eqref{eq:proofProposition2516}. We have
    \begin{equation}
        \lim_{N \to \infty} s_N \sum_{k=1}^{\infty} e^{s_N k} = \lim_{N \to \infty} s_N \frac{1}{e^{-s_N} - 1} = -1.
        \label{eq:proofProposition2527}
    \end{equation}
    Recall that $t_N = |s_N|$. For any fixed $\epsilon > 0$, 
    \begin{align}
        \left| \sum_{k=1}^{\infty} e^{s_N k} R(\beta k) \right| &\leq \left( \sup_{1 \leq k \leq 1/(\epsilon t_N)} |R(\beta k)| \right) \sum_{1 \leq k \leq 1/(\epsilon t_N)} 1 + C \sum_{k > 1/(\epsilon t_N)} e^{s_N k} \nonumber \\
        &\leq \left( \sup_{1 \leq k \leq 1/(\epsilon t_N)} |R(\beta k)| \right) \frac{1}{\epsilon t_N} + \frac{C}{t_N} e^{-\epsilon^{-1}}.
        \label{eq:proofProposition2528}
    \end{align}
    The supremum on the right-hand side converges to $0$ for $N \to \infty$ because $\beta/t_N \to 0$, and hence
    \begin{equation}
        \lim_{N \to \infty} s_N \sum_{k=1}^{\infty} e^{s_N k} R(\beta k) = 0.
        \label{eq:proofProposition2529}
    \end{equation}
    Together with \eqref{eq:proofProposition2526} and \eqref{eq:proofProposition2527}, this shows
    \begin{equation}
        \lim_{N \to \infty} (-\beta s_N) \sum_{j=1}^{\infty} \mathrm{Var}(Y_j) = L.
        \label{eq:proofProposition2530}
    \end{equation}

    Next we derive an estimate on $\sum_{j=1}^{\infty} \mathbb{E}(|Y(s_N) - \mathbb{E}(s_N)|^3)$. The bound $|a-b|^3 \leq 4(a^3 + b^3)$ for $a,b \geq 0$ and an application of Jenson's inequality show
    \begin{equation}
        \mathbb{E} \left( | Y_j - \mathbb{E}(Y_j) |^3 \right) \leq 4 \mathbb{E} ( Y_j^3 + (\mathrm{E}[Y_j])^3 ) \leq 8 \mathbb{E} ( Y_j^3 ).
        \label{eq:proofProposition2531}
    \end{equation}
    Using this bound and $\sum_{k=1}^{\infty} k^3 q^k = q(1+4q+q^2)/(1-q)^4$ for $0 \leq q < 1$, we find
    \begin{align}
        \sum_{j=1}^{\infty} \mathbb{E}(|Y_j - \mathbb{E}[Y_j]|^3) &\leq 8 \sum_{j=1}^{\infty} (\mathrm{E}[Y_j^3]) = 8 \sum_{j=1}^{\infty} \left[ (1-e^{-\beta E_j + s_N}) \sum_{k=1}^{\infty} k^3 e^{-(\beta E_j - s_N)k} \right] \nonumber \\
        &=8 \sum_{j=1}^{\infty} \frac{e^{-(\beta E_j - s_N)} \left[ 1 + 4 e^{-(\beta E_j - s_N)} + e^{-(\beta E_j - s_N)2} \right]}{[1-e^{-(\beta E_j - s_N)}]^3}.
        \label{eq:proofProposition2532}
    \end{align}
    We also have $\sum_{k=1}^{\infty} k^2 q^k = q(1+q)/(1-q)^3$ and $1+4q+q^2 \leq 3(1+q)$ for $0 \leq q < 1$, and hence
    \begin{align}
        \frac{1}{\sum_{k=1}^{\infty} k^2 e^{-(\beta E_j - s_N)k}} \frac{e^{-(\beta E_j - s_N)} \left[ 1 + 4 e^{-(\beta E_j - s_N)} + e^{-(\beta E_j - s_N)2} \right]}{(1-e^{-(\beta E_j - s_N)})^3} &= \frac{1+4 e^{-(\beta E_j - s_N)} + e^{-(\beta E_j - s_N)2} }{1+ e^{-(\beta E_j - s_N)}} \nonumber \\ 
        &\leq 3.
        \label{eq:proofProposition2533}
    \end{align}
    In combination with \eqref{eq:proofProposition2532} and an application of Fubini's theorem, this implies
    \begin{equation}
        \sum_{j=1}^{\infty} \mathbb{E}(|Y_j - \mathbb{E}[Y_j]|^3) \leq 24 \sum_{k=1}^{\infty} k^2 e^{s_N k} \left( \sum_{j=1}^{\infty} e^{-\beta E_j k} \right).
        \label{eq:proofProposition2534}
    \end{equation}
    We use $E_j \geq K j$ for some $K>0$ to check that there exists a constant $C>0$ such that $\sum_{j=1}^{\infty} e^{-\beta E_j k} \leq C/(\beta k)$ holds for all $k \geq 1$. Insertion of this bound into \eqref{eq:proofProposition2534} and an application of the inequality $\sinh(x) \geq x$ for $x \geq 0$ yield
    \begin{equation}
        \sum_{j=1}^{\infty} \mathbb{E}(|Y_j - \mathbb{E}[Y_j]|^3) \leq \frac{C}{\beta} \sum_{k=1}^{\infty} k e^{s_N k} = \frac{C}{\beta} \frac{1}{4 \sinh^2\left( \frac{-s_N}{2} \right)} \leq \frac{C}{\beta s_N^2}.
        \label{eq:proofProposition2535}
    \end{equation}

    In particular,
    \begin{equation}
        \lim_{N \to \infty} \frac{\sum_{j=1}^{\infty} \mathbb{E}(|Y_j - \mathbb{E}[Y_j]|^3)}{\left( \sum_{j=1}^{\infty} \mathrm{Var}[Y_j] \right)^{3/2}} = 0.
        \label{eq:proofProposition2536}
    \end{equation}
    To obtain this, we also used $\beta/|s_N| \to 0$ for $N \to \infty$. Eq.~\eqref{eq:proofProposition2536} is the condition that is needed for \cite[Theorem~8.4.1]{Borovkov2013} to hold. An application of this theorem proves the claim.
\end{proof}

We recall that the law $\mu_{Z_j}$ of $Z_j$ is absolutely continuous with respect to the law $\mu_{Y_j}$ of $Y_j$ with Radon-Nikodym derivative given by the right-hand side of \eqref{eq:RadonNykodym}. In our case we need to replace $s$ by $s_N$. Hence, we have
\begin{align}
\mathbb{P}\left(M \in B \right)&=\exp\left( \frac{[\log(b/\beta)]^{\eta b/b_\mathrm{c}}}{\beta^{c}} g_N(s_N) \right) \mathbb{E}\left[e^{-s_N \left(Y-\mathbb{E}[Y] \right)}
\mathbf{1}_{\{Y \in B\}}\right]
\label{eq:proofProposition2537}
\end{align}
for any Borel set $B \subseteq \mathbb{R}$.

The arguments that lead to \eqref{eq:AlgebraLowerBoundRightTail} also show
\begin{align}
    \log \mathbb{P}(Z_0>(1+\eta+\epsilon_N)\mathbb{E}(Z_0)) \geq& \frac{[b \log(b/\beta)]^{\eta b/b_\mathrm{c}}}{\beta^{c}} g_N(s_N) \nonumber \\
    &+ \log \mathbb{E}\left[ e^{-s_N(Y(s_N) - \mathbb{E}[Y(s_N)] ) )} \mathbf{1}_{\{ \mathbb{E}[Y(s_N)] - \sqrt{\mathrm{Var}(Y(s_N))} < Y(s_N) < \mathbb{E}[Y(s_N)] \} } \right] \nonumber \\
    \geq& \frac{[b \log(b/\beta)]^{\eta b/b_\mathrm{c}}}{\beta^{c}} g_N(s_N) + s_N \sqrt{\mathrm{Var}(Y(s_N))}  \label{eq:proofProposition2538} \\
    &+ \log \mathbb{P} \left( \mathbb{E}[Y(s_N)] - \sqrt{\mathrm{Var}(Y(s_N))} < Y(s_N) < \mathbb{E}[Y(s_N)] \right). \nonumber
\end{align}
We multiply both sides with the relevant prefactor, apply Lemmas~\ref{lem:asymptoticsOfsN} and \ref{lem:CLT}, and use \eqref{eq:proofProposition2523} and \eqref{eq:proofProposition2530}
to see that
\begin{align}
    \liminf_{N \to \infty} \frac{\beta^{c}}{[b \log(b/\beta)]^{\eta b/b_\mathrm{c}}} \log \mathbb{P}(Z_0>(1+\eta+\epsilon_N)\mathbb{E}(Z_0)) \geq -L e^{-(1+\eta) C_{\mathrm{E}}}. 
    \label{eq:proofProposition2539}
\end{align}
Together with \eqref{eq:proofProposition2524}, \eqref{eq:proofProposition2539} proves Proposition~\ref{prop:right tail1} in the case $0 < \eta < [(b/b_{\mathrm{c}})^\alpha - 1]^{-1}$. The proof in the case $\eta > [(b/b_{\mathrm{c}})^\alpha - 1]^{-1}$ is straightforward, and therefore left to the reader.
\subsection{Proof of Theorem~\ref{thm:main}}
\label{sec:proofTheorem1}
In this section we use Lemma~\ref{lemma~conditon} and Propositions~\ref{prop:left tail}--\ref{prop:right tail1} to obtain a proof of Theorem~\ref{thm:main}. 
\subsubsection*{Left tail}
From Proposition~\ref{prop:left tail} with the choice $\eta = 1$ and the identity $Z_0 = N - M$ we know that
\begin{equation}
\lim_{N\to\infty} \frac{1}{N\beta}\log \mathbb{P}(M>N) = \left[\left(\frac{b_c}{b}\right)^\alpha-1\right]E_1.
\label{Asymptotic of P(M>N)}
\end{equation}

We recall the definition of $N_0$ above \eqref{eq:beta} and pick some $0 < \eta < 1$. An application of Lemma~\ref{lemma~conditon} allows us to write
\begin{align}
    \mathbb{P}^{\mathrm{can}}\bigl(N_0<(1-\eta)\mathbb{E}(N_0)\bigr) &= \mathbb{P}\bigl(Z_0<(1-\eta)\mathbb{E}(Z_0 \mid M \le N)\mid M \le N \bigr) \nonumber \\
    &\leq \frac{\mathbb{P}\bigl(Z_0<(1-\eta)\mathbb{E}(Z_0\mid M\le N)\bigr)}{1-\mathbb{P}(M>N)}.
    \label{eq:proofMainThm1}
\end{align}
Let us compare $\mathbb{E}(Z_0 \mid M \le N)$ and $\mathbb{E}(Z_0)$. Since $\{M \leq N\} = \{ Z_0 \geq 0\}$ we have 
\begin{align}
    0 \leq \mathbb{E}(Z_0 \mid M \le N) - \mathbb{E}(Z_0) &= \left[ \frac{P(M > N)}{1 - P(M>N)} \right] \mathbb{E}[Z_0 \mathbf{1}_{Z_0 > 0}] + \mathbb{E}[(M-N)\mathbf{1}_{M > N}]  \nonumber \\
    &\leq 2 \mathbb{E}(Z_0) P(M > N) + \sqrt{ \mathrm{Var}(M) P(M>N) }
    \label{eq:proofMainThm2}
\end{align}
provided $N$ is sufficiently large. To come to the last line, we used \eqref{Asymptotic of P(M>N)}, which guarantees $\mathbb{P}(M>N) \leq 1/2$ for $N$ large enough, $\mathrm{E}(M) \leq N$, and the Cauchy-Schwarz inequality. If we additionally use $\mathbb{E}(Z_0) \sim N$, \eqref{Asymptotic of P(M>N)}, and $\mathrm{Var}(M) \leq C/\beta^{1+\alpha}$ (see proofs of Lemma~3.2 and 3.3 in \cite{ChatDiac2014}) we find 
\begin{equation}
    \frac{\mathbb{E}(Z_0 \mid M \le N) - \mathbb{E}(Z_0)}{\mathbb{E}(Z_0)} \leq \frac{C}{\beta^{(1+\alpha)/2} N} \exp \left( \frac{\beta N}{2} \left[\left(\frac{b_c}{b}\right)^\alpha-1\right]E_1(1+o(1)) \right).
    \label{eq:proofMainThm2b}
\end{equation}
Using \eqref{eq:beta} one checks that the right-hand side goes to zero for $N \to \infty$. 

Pick $\epsilon > 0$ such that $\eta(1+\epsilon) > \epsilon$ and choose $N$ large enough such that $\mathbb{E}(Z_0 \mid M \le N) - \mathbb{E}(Z_0) \leq \epsilon \mathbb{E}(Z_0)$ holds. This is possible because of \eqref{eq:proofMainThm2b}. An application of \eqref{eq:proofMainThm1} shows
\begin{equation}
    \log \mathbb{P}^{\mathrm{can}}\bigl(N_0<(1-\eta)\mathbb{E}(N_0)\bigr) \leq \log \mathbb{P}\bigl(Z_0<(1-\eta)(1 + \epsilon)\mathbb{E}(Z_0)\bigr) - \log \left(1-\mathbb{P}(M>N) \right),
    \label{eq:proofMainThm3}
\end{equation}
and hence
\begin{equation}
    \limsup_{N \to \infty} \frac{1}{\beta N} \log \mathbb{P}^{\mathrm{can}}\bigl(N_0<(1-\eta)\mathbb{E}(N_0)\bigr) \leq \left( \eta(1+\epsilon) - \epsilon \right) \left[\left(\frac{b_{\mathrm{c}}}{b}\right)^{\alpha}-1 \right]E_1
    \label{eq:proofMainThm4}
\end{equation}
holds by Proposition~\ref{prop:left tail} and \eqref{Asymptotic of P(M>N)}. We take the limit $\epsilon \to 0$ on both sides to obtain \eqref{eq:proofMainThm4} with $\epsilon = 0$. It remains to prove a matching lower bound. 

For any two events $A,B$ and any $0 < \epsilon < 1$ we have
\begin{equation}
    P( A \cap B ) \geq (1-\epsilon)P(A) - \frac{1}{4 \epsilon} P(B^{\mathrm{c}}),
    \label{eq:proofMainThm5a}
\end{equation}
where $B^{\mathrm{c}}$ denotes the complement of $B$. We use \eqref{eq:proofMainThm1}, the first inequality in \eqref{eq:proofMainThm2} and \eqref{eq:proofMainThm5a} to estimate
\begin{align}
    \mathbb{P}^{\mathrm{can}}\bigl(N_0<(1-\eta)\mathbb{E}(N_0)\bigr) &\geq \frac{ \mathbb{P}\bigl(Z_0<(1-\eta)\mathbb{E}(Z_0), M \leq N\bigr) }{ 1-\mathbb{P}(M>N) } \nonumber \\
    &\geq (1-\epsilon)\mathbb{P}\bigl(Z_0<(1-\eta)\mathbb{E}(Z_0) \bigr)  - \frac{1}{4 \epsilon} \mathbb{P}(M > N) \label{eq:proofMainThm5} \\
    &= (1-\epsilon)\mathbb{P}\bigl(Z_0<(1-\eta)\mathbb{E}(Z_0) \bigr) \left[ 1 - \frac{P(M>N)}{4 \epsilon(1-\epsilon) \mathbb{P}\bigl(Z_0<(1-\eta)\mathbb{E}(Z_0) \bigr)} \right].
    \nonumber
\end{align}
Together with Proposition~\ref{prop:left tail}, \eqref{Asymptotic of P(M>N)} and the conditions $0 < \eta < 1$ and $b>b_{\mathrm{c}}$, this implies
\begin{equation}
    \liminf_{N \to \infty} \frac{1}{\beta N} \log \mathbb{P}^{\mathrm{can}}\bigl(N_0<(1-\eta)\mathbb{E}(N_0)\bigr) \geq \eta \left[\left(\frac{b_{\mathrm{c}}}{b}\right)^{\alpha}-1 \right]E_1.
    \label{eq:proofMainThm6}
\end{equation}
In combination, \eqref{eq:proofMainThm4} and \eqref{eq:proofMainThm6} prove part~(a) of Theorem~\ref{thm:main}.
\subsubsection*{Right tail}
Assume that $0 < \eta < [(b/b_{\mathrm{c}})^\alpha - 1]^{-1}$. Using $\{ Z_0 \geq 0 \} = \{ M \leq N\}$ and the first inequality in \eqref{eq:proofMainThm2}, we check that $\mathbb{E}(Z_0\mid M\le N) \geq 0$. In particular,
\begin{equation}
    \left\{Z_0>(1+\eta)\mathbb{E}(Z_0\mid M\le N)\right\} \subset \{M\le N\}.
    \label{eq:proofMainThm7}
\end{equation}
But this implies 
\begin{align}
    \mathbb{P}^{\mathrm{can}}\bigl(N_0>(1+\eta)\mathbb{E}(N_0)\bigr) &= \mathbb{P}\bigl(Z_0>(1+\eta)\mathbb{E}(Z_0 \mid M \le N)\mid M \le N \bigr) \nonumber \\
    &= \frac{\mathbb{P}\bigl(Z_0>(1+\eta)\mathbb{E}(Z_0 \mid M \le N) \bigr)}{1 - \mathbb{P}(M > N)}.
    \label{eq:proofMainThm8}
\end{align}

Assume first that $\alpha > 1$. We use \eqref{Asymptotic of P(M>N)}, the first bound in \eqref{eq:proofMainThm2}, and Proposition~\ref{prop:right tail} to see that
\begin{equation}
    \limsup_{N \to \infty} \frac{1}{N} \log\mathbb{P}^{\mathrm{can}}\bigl(N_0>(1+\eta)\mathbb{E}(N_0)\bigr) \leq \inf_{s < 0} \phi_{\eta,b}(s).
    \label{eq:proofMainThm8}
\end{equation}

For the reverse bound we pick $\epsilon > 0$ such that $\eta + \epsilon(1+\eta) < [(b/b_{\mathrm{c}})^\alpha - 1]^{-1}$ holds. We also choose $N$ large enough such that $\mathbb{E}(Z_0 \mid M \le N) - \mathbb{E}(Z_0) \leq \epsilon \mathbb{E}(Z_0)$. Applications of \eqref{Asymptotic of P(M>N)} and Proposition~\ref{prop:right tail} show
\begin{equation}
    \liminf_{N \to \infty} \frac{1}{N} \log\mathbb{P}^{\mathrm{can}}\bigl(N_0>(1+\eta)\mathbb{E}(N_0)\bigr) \geq \inf_{s<0} \phi_{\eta(\epsilon),b}(s),
    \label{eq:proofMainThm9}
\end{equation}
where $\eta(\epsilon) = \eta + \epsilon(1+\eta)$. Below \eqref{eq:lowerBoundProf241} we concluded that $\phi_{\eta,b}$ has a unique minimizer that we denote by $s(\eta)$. It is not difficult to check that $\sup_{0 \leq \delta \leq \epsilon}|s(\eta(\epsilon))| < + \infty$, and hence
\begin{align}
    \inf_{s<0} \phi_{\eta(\epsilon),b}(s) &= \phi_{\eta(\epsilon),b}(s(\eta(\epsilon)) = \phi_{\eta,b}(s(\eta(\epsilon)) + \epsilon s(\eta(\epsilon)) \left[ 1 - \left( \frac{b_{\mathrm{c}}}{b} \right)^{\alpha} \right] \nonumber \\
    & \geq \inf_{s < 0} \phi_{\eta,b}(s) - C \epsilon.
    \label{eq:proofMainThm10}
\end{align}
In combination, \eqref{eq:proofMainThm8}, \eqref{eq:proofMainThm9}, and \eqref{eq:proofMainThm10}  prove part~(b) of Theorem~\ref{thm:main} in the case $0 < \eta < [(b/b_{\mathrm{c}})^\alpha - 1]^{-1}$. The proof in the case $\eta > [(b/b_{\mathrm{c}})^\alpha - 1]^{-1}$ is straightforward, and therefore left to the reader.

It remains to consider the case $\alpha = 1$. The arguments that led to \eqref{eq:proofMainThm8} work in the same way for $\alpha = 1$, and we therefore focus on the lower bound. 

We would like to replace $\mathbb{E}(Z_0 | M \leq N)$ by $(1+\epsilon) \mathbb{E}(Z_0)$ but this is not possible because for $\alpha = 1$ also the exponential rate depends on $\eta$. The way out is to let $\epsilon$ depend on $N$ and to quantify the error that we make. To that end, we use \eqref{eq:beta} and \eqref{eq:proofMainThm2b} to see that
\begin{align}
    \frac{\mathbb{E}(Z_0 \mid M \le N) - \mathbb{E}(Z_0)}{\mathbb{E}(Z_0)} \leq \frac{C}{b \ln(N)} N^{\frac{b}{2}\left[ \frac{b_{\mathrm{c}}}{b} - 1 \right] E_1(1+o(1))}.
    \label{eq:proofMainThm11}
\end{align}
We choose $D > C$, where $C$ denotes the constant on the right-hand side of \eqref{eq:proofMainThm11}, and denote 
\begin{equation}
    \epsilon_N = \frac{D}{b \ln(N)} N^{\frac{b}{4}\left[ \frac{b_{\mathrm{c}}}{b} - 1 \right] E_1}.
    \label{eq:definitionEpsilonN}
\end{equation}
Using \eqref{eq:proofMainThm8} and Proposition~\ref{prop:right tail1}, we check that
\begin{align}
    &\liminf_{N \to \infty} \frac{\beta^{\eta [b/b_{\mathrm{c}}-1]}}{[b \log(b/\beta)]^{\eta b/b_\mathrm{c}}} \log \mathbb{P}^{\mathrm{can}}\bigl(N_0>(1+\eta)\mathbb{E}(N_0)\bigr) \nonumber \\
    &\hspace{2cm}\geq \liminf_{N \to \infty} \frac{\beta^{\eta [b/b_{\mathrm{c}}-1]}}{[b \log(b/\beta)]^{\eta b/b_\mathrm{c}}} \log \mathbb{P}\bigl(Z_0>(1+\eta + \epsilon_N)\mathbb{E}(Z_0 ) \bigr) \nonumber \\
    &\hspace{2cm}= -L e^{-(1+\eta) C_{\mathrm{E}}},
    \label{eq:proofMainThm12}
\end{align}
which proves part~(c) of Theorem~\ref{thm:main} in the case $0 < \eta < [(b/b_{\mathrm{c}})^\alpha - 1]^{-1}$. The proof in the case $\eta > [(b/b_{\mathrm{c}})^\alpha - 1]^{-1}$ is straightforward, and therefore left to the reader.

\vspace{0.5cm}

    \textbf{Acknowledgments.} A. D. gratefully acknowledges support through the NSF Grant DMS-2555747. It is a pleasure for A. D. to acknowledge stimulating discussions with Simone Rademacher. X. L. gratefully acknowledges support through the NSF Grant DMS-2307093.
    
\vspace{0.5cm}

    \textbf{AI Disclosure.} GPT-5.6 Sol has been used to improve the presentation, to help with computations, and to check the correctness of proofs.

\vspace{0.5cm}   

    \textbf{Data availability statement.} No datasets were generated or analysed during the current study.

\bibliographystyle{siam}

\begin{thebibliography}{10}
\bibitem{Andetal1995} M.H. Anderson, J.R. Ensher, M.R. Matthews, C.E. Wieman, E.A. Cornell, \textit{Observation of bose-einstein condensation in a dilute atomic vapor}, Science \textbf{269}, 198 (1995)

\bibitem{BastBocCenaDeu2025} G. Basti, C. Boccato, S. Cenatiempo, A. Deuchert, \textit{A new upper bound on the specific free energy of dilute Bose gases}, arXiv:2507.20877 [math-ph] (2025)

\bibitem{BehrBreRade2025} N. Behrmann, C. Brennecke and S. Rademacher, \textit{Exponential control of excitations for trapped BEC in the Gross--Pitaevskii regime}, Lett. Math. Phys. \textbf{115}, 91 (2025)

\bibitem{BenPoSchl2016} N. Benedikter, M. Porta, B. Schlein, \textit{Eﬀective Evolution Equations from Quantum Dynamics}, Springer, Berlin (2016)

\bibitem{BergLewisPule1988} M. van den Berg, J. T. Lewis, and J. V. Pulé, \textit{The large deviation principle and some models of an
interacting boson gas}, Commun. Math. Phys. \textbf{118}, 61–85  (1988)

\bibitem{Borovkov2013} A. A. Borovkov, \textit{Probability Theory}, Springer, London, 2013


\bibitem{BruZag2008} J.-B. Bru, V. A. Zagrebnov, \textit{Large Deviations in the Superstable Weakly Imperfect Bose-Gas}, J. Stat. Phys. \textbf{133}, 379 (2008)

\bibitem{BocDeuSto2024} C. Boccato, A. Deuchert, D. Stocker, \textit{Upper bound for the grand canonical free energy of the Bose
gas in the Gross–Pitaevskii limit}, SIAM J. Math. Anal. \textbf{56}, no. 2, 2611-2660 (2024)

\bibitem{BufPulet1983} E. Buffet, J. V. Pulé, \textit{Fluctuation properties of the imperfect Bose gas}, J. Math. Phys. \textbf{24}, 1608–1616 (1983)

\bibitem{Bureikatal2023} J.-P. Bureik, G. Hercé, M. Allemand, 
T. Roscilde, D. Clément, \textit{Suppression of Bogoliubov momentum
pairing and emergence of non-Gaussian correlations in ultracold interacting Bose gases}, Nat. Phys. \textbf{21}, 57–62 (2025) 

\bibitem{CapDeu2025} M. Caporaletti, A. Deuchert, \textit{Upper bound for the grand canonical free energy of the Bose gas in
the Gross-Pitaevskii limit for general interaction potentials}, Ann. Henri Poincaré \textbf{26}, 3767–3827 (2025)

\bibitem{DeuCapSchl2024} M. Caporaletti, A. Deuchert, B. Schlein, \textit{Dynamics of mean-field bosons at positive temperature}, Ann. Inst. Henri Poincaré (C) Anal. Non Linéaire \textbf{41}, no. 4, pp. 995 (2024)

\bibitem{ChatDiac2014} S. Chatterjee, P. Diaconis, \textit{Fluctuations of the Bose-Einstein condensate}, J. Phys. A: Math. Theor. \textbf{47}, 085201 (2014)

\bibitem{Christetal2021} M. B. Christensen, T. Vibel, A. J. Hilliard, M. B. Kruk, K. Pawłowski, D. Hryniuk, K. Rzażewski, M. A. Kristensen,
J. J. Arlt, \textit{Observation of Microcanonical Atom Number Fluctuations in a Bose-Einstein Condensate}, Phys. Rev.
Lett. \textbf{126}, 153601 (2021)

\bibitem{PraCara2013} P. Dai Pra, F. Caravenna, \textit{Probabilità} (In Italian), Springer (2013)

\bibitem{DavisEtal1995} K.B. Davis, M.-O. Mewes, M.R. Andrews, N.J. van Druten, D.S. Durfee, D.M. Kurn, W. Ketterle, \textit{Bose-Einstein Condensation in a Gas of Sodium Atoms}, Phys. Rev. Lett. \textbf{75}, 3969 (1995)

\bibitem{DeuNamNap2025} A. Deuchert, P. T. Nam, M. Napiórkowski, \textit{The Gibbs state of the mean-field Bose gas}, arXiv:2501.19396 [math-ph] (2025)

\bibitem{DeuSeiYng2019} A. Deuchert, R. Seiringer, J. Yngvason, \textit{Bose–Einstein condensation in a dilute, trapped gas at
positive temperature}, Commun. Math. Phys. \textbf{368}, 723-776 (2019)

\bibitem{DeuSei2020} A. Deuchert, R. Seiringer, \textit{Gross–Pitaevskii limit of a homogeneous Bose gas at positive temperature},
Arch. Ration. Mech. Anal. \textbf{236}, 1217–1271 (2020)

\bibitem{FroeKnowSchlSoh2022} J. Fröhlich, A. Knowles, B. Schlein, V. Sohinger, \textit{The mean–field limit of quantum Bose gases at
positive temperature}, J. Amer. Math. Soc. \textbf{35}, no. 4, 955 (2022)

\bibitem{FroeKnowSchlSoh2025} J. Fröhlich, A. Knowles, B. Schlein, V. Sohinger, J. Eur. Math. Soc. \textbf{27}, 4399–4468 (2025)

\bibitem{GalLebMastr2002} G. Gallavotti, J. L. Lebowitz, V. Mastropietro, \textit{Large deviations in rarefied quantum gases}, J. Stat. Phys. 108(5–6), 831–861 (2002)

\bibitem{HerceEtAl2023} G.~Hercé, J.~P.~Bureik, A.~Ténart, A.~Aspect, A.~Dareau, D.~Clément, \textit{Full counting statisics of interacting lattics gases after an expansion: The role of condensate depletion in many-body coherence}, Phys. Rev. Research \textbf{5}, 012037 (2023)

\bibitem{KirkRadeSchlein2021} K. Kirkpatrick, S. Rademacher, B. Schlein, \textit{A large deviation principle for many–body quantum dynamics}, Ann. Henri Poincaré \textbf{22}, 2595–2618 (2021)

\bibitem{LamiraultEtAl2025} C.~Lamirault, R.~Dias, C.~Leprince, C.~I.~Westbrook, D.~Clément, D.~Boiron, \textit{Quantifying Two-Mode Entanglement of Bosonic Gaussian States from Their Full Counting Statistics}, Phys. Rev. Lett. \textbf{135}, 100201 (2025)

\bibitem{LebLenciSpohn2000} J. L. Lebowitz, M. Lenci, H. Spohn, \textit{Large deviations for ideal quantum systems}, J. Math. Phys. \textbf{41},
1224–1243 (2000)

\bibitem{LieSeiSolYng2005} E.H. Lieb, R. Seiringer, J. P. Solovej, J. Yngvason, \textit{The Mathematics of the Bose Gas and its
Condensation}, Birkhäuser, Basel (2005)

\bibitem{LewNamRou2021} M. Lewin, P. T. Nam and N. Rougerie, \textit{Classical field theory limit of many-body quantum Gibbs states in 2D and 3D}, Invent. Math. \textbf{224} (2021), 315-444

\bibitem{Nap2023} M. Napiórkowski, \textit{Dynamics of interacting bosons: a compact review}, In: Density Functionals for
Many-Particle Systems - Mathematical Theory and Physical Applications of Effective Equations, World Scientific (2023)

\bibitem{NamRade2025} P. T. Nam, S. Rademacher, \textit{Exponential bounds of the condensation for dilute Bose gases}, Trans. Amer. Math. Soc. \textbf{378}, 3229-3278 (2025)

\bibitem{NamZhuZhu2025} P. T. Nam, R. Zhu, X. Zhu, \textit{$\Phi^4_3$ Theory from many-body quantum Gibbs states}, arXiv:2502.04884 [math-ph] (2025)

\bibitem{Radem2025} S. Rademacher, \textit{Large deviations for the ground state of weakly interacting Bose gases}, Ann. Henri Poincaré
\textbf{26}, 1239–1289 (2025)

\bibitem{RaSei-2022} S. Rademacher, R. Seiringer, \textit{Large deviation estimates for weakly interacting bosons}, J. Stat. Phys. \textbf{188} (9) (2022)

\bibitem{RS4} M. Reed, B. Simon, \textit{Methods of modern mathematical physics IV, Analysis of Operators}, Academic Press, San Diego (1980)

\bibitem{Roug2015} N. Rougerie, \textit{De Finetti theorems, mean-field limits and Bose-Einstein condensation}, arXiv:1506.05263, Lecture notes (2015)

\bibitem{SakKase2016} K. Sakmann, M. Kasevich, \textit{Single-shot simulations of dynamic quantum many-body systems}, Nat. Phys. \textbf{12}, 451 (2016)

\bibitem{Solovej2025} J. P. Solovej, \textit{Mathematical physics of dilute Bose gases}, Comptes Rendus. Physique, Volume \textbf{26}, 339-348 (2025)

\bibitem{TenartEtAl2021} A.~Tenart, G.~Hercé, J.~P.~Bureik, A.~Dareau, D.~Clément, \textit{Observation of pairs of atoms at opposite momenta in an equilibrium interacting Bose gas}, Nat. Phys. \textbf{17}, 1364–1368 (2021)

\bibitem{Varadhan1984} S. R. S. Varadhan,
\textit{Large Deviations and Applications}, SIAM, Philadelphia, 1984.

\end{thebibliography}

\vspace{0.5cm} 

\noindent (Andreas Deuchert) Department of Mathematics, Virginia Tech \\ 
225 Stanger Street, Blacksburg, VA 24060-1026, USA \\ 
E-mail address: \texttt{andreas.deuchert@vt.edu} \\

\vspace{0.2cm} 

\noindent (Xuanyu Li) Department of Mathematics, Virginia Tech \\ 
225 Stanger Street, Blacksburg, VA 24060-1026, USA \\ 
E-mail address: \texttt{xuanyu@vt.edu} \\

\end{document}